\documentclass[11pt]{article}

\usepackage{a4wide}
\usepackage{authblk}
\usepackage{amsmath,amssymb}
\usepackage{amsthm}
\usepackage[dvipdfmx]{graphicx}
\usepackage{here}
\usepackage{url}

\pdfoutput=1

\usepackage{thmtools}

\newtheorem{theorem}{Theorem}
\newtheorem{lemma}{Lemma}

\theoremstyle{definition}
\newtheorem{definition}{Definition}

\newtheorem{example}{Example}

\usepackage[final,mode=multiuser]{fixme}
\fxsetup{theme=color}
\definecolor{fxtarget}{rgb}{0.0000,0.0000,0.7000}
\fxuseenvlayout{colorsig}
\fxusetargetlayout{color}
\FXRegisterAuthor{rh}{ahr}{RH}
\FXRegisterAuthor{hf}{ahf}{HF}
\FXRegisterAuthor{si}{asi}{SI}
\fxsetface{env}{}

\newcommand{\Prefix}{\mathsf{Prefix}}
\newcommand{\Substr}{\mathsf{Substr}}
\newcommand{\Suffix}{\mathsf{Suffix}}

\newcommand{\CDAWG}{\mathsf{CDAWG}}

\newcommand{\lrep}{\mathsf{lexp}}
\newcommand{\rrep}{\mathsf{rexp}}

\newcommand{\LeftM}{\mathsf{LeftM}}
\newcommand{\RightM}{\mathsf{RightM}}
\newcommand{\M}{\mathsf{M}}
\newcommand{\D}{\mathsf{d}}

\newcommand{\size}{\mathsf{e}}

\newcommand{\MSIns}{\mathsf{MS}_{\mathrm{Ins}}}
\newcommand{\MSDel}{\mathsf{MS}_{\mathrm{Del}}}
\newcommand{\MSSub}{\mathsf{MS}_{\mathrm{Sub}}}

\newcommand{\N}{\mathsf{N}}
\newcommand{\None}{\mathsf{N}_{\mathrm{1}}}
\newcommand{\Ntwo}{\mathsf{N}_{\mathrm{2}}}
\newcommand{\Nthree}{\mathsf{N}_{\mathrm{3}}}
\newcommand{\Nv}{\mathsf{N}_{\mathrm{3B}}}
\newcommand{\Q}{\mathsf{Q}}
\newcommand{\Qn}{\mathsf{Q}_{\mathrm{1}}}
\newcommand{\Qnotn}{\mathsf{Q}_{\mathrm{2}}}
\newcommand{\Nnotv}{\mathsf{N}_{\mathrm{3A}}}

\newcommand{\ed}{\mathsf{ed}}

\newcommand{\polylog}{\mathrm{polylog}}

\begin{document}

\title{Constant sensitivity on the CDAWGs}

\author[1]{Rikuya~Hamai}
\author[1]{Hiroto~Fujimaru}
\author[2]{Shunsuke~Inenaga}

\affil[1]{Department of Information Science and Technology, Kyushu University, Japan}
\affil[2]{Department of Informatics, Kyushu University, Japan}

\date{}
\maketitle

\begin{abstract}
\emph{Compact directed acyclic word graphs} (\emph{CDAWGs}) [Blumer et al. 1987] are a fundamental data structure on strings with applications in text pattern searching, data compression, and pattern discovery. Intuitively, the CDAWG of a string $T$ is obtained by merging isomorphic subtrees of the suffix tree [Weiner 1973] of the same string $T$, and thus CDAWGs are a compact indexing structure. In this paper, we investigate the sensitivity of CDAWGs when a single character edit operation is performed at an arbitrary position in $T$.
We show that the size of the CDAWG after an edit operation on $T$
is asymptotically at most 8 times larger than the original CDAWG before the edit.
\end{abstract}

\section{Introduction}

\emph{Compact directed acyclic word graphs} (\emph{CDAWGs})~\cite{Blumer1987} are a fundamental data structure on strings with applications in text pattern searching, data compression, and pattern discovery.
Intuitively, the CDAWG of a string $T$ (denoted $\CDAWG(T)$) is a minimal partial DFA that is obtained by merging isomorphic subtrees of the suffix tree~\cite{Weiner1973} of the same string $T$.
CDAWGs permit pattern matching in optimal $O(m + occ)$ time
for a pattern $P$ of length $m$ when $P$ occurs $occ$ times in $T$.
In practice, CDAWGs enjoy applications in natural language processing~\cite{Takeda2000} and in analysis of text generated by language models (LMs)~\cite{MerrillSE24}.
In a more theoretical perspective, CDAWGs are used as a space-efficient data structure that allows for optimal-time detection of ``unusual words'' (such as minimal absent words (MAWs)~\cite{Crochemore1998MAWdefinition} and minimal unique substrings (MUSs)~\cite{Ilie2011MUS}) from the input string~\cite{BelazzouguiC17,InenagaMAFF24,MienoI25}.

A substring $w$ of $T$ which occurs at least twice in $T$ is said to be a \emph{maximal repeat} of $T$,
if extending $w$ to the left or to the right in $T$ decreases
the number of occurrences in $T$.
It is known that there is a one-to-one correspondence between the internal nodes of $\CDAWG(T)$ and the maximal repeats in $T$.
Also, there is a one-to-one correspondence between the out-edges of
the internal nodes of $\CDAWG(T)$ and
the right-extensions of the maximal repeats in $T$.
Let $\size$ denote the size (i.e. the number of edges) of
the CDAWG for the input string.
It is known that $\size \leq 2n-2$ holds~\cite{Blumer1987} for any strings of length $n$,
and $\size$ can be much smaller for some highly repetitive strings:
$\size \in \Theta(\log n)$ holds for
Fibonacci words, Standard Sturmian words, and Thue-Morse words of length $n$~\cite{Rytter06,BaturoPR09,RadoszewskiR12}.
This contrasts to the suffix tree and the (uncompacted) directed acyclic word graph (DAWG)~\cite{Blumer1985} each requiring $\Theta(n)$ space
for any string of length $n$.
CDAWGs can thus be regarded as a compressed text indexing structure
which can be stored in $O(\size)$ space~\cite{BelazzouguiC17,Inenaga24} without explicitly storing the string.
In addition, a grammar compression of size $O(\size)$
based on the CDAWG exists~\cite{BelazzouguiC17}.

The \emph{sensitivity} of string compressors,
first proposed by Akagi et al.~\cite{AkagiFI2023},
measures how much a single-character-wise edit operation on the input string
can increase the size of the compressed string,
which is formalized as follows:
Let $C$ be a compression algorithm and let $C(T)$ denote
the size of the output of $C$ applied to the input string $T$.
The worst-case \emph{multiplicative sensitivity} of $C$ 
is defined by
$$\max_{T \in \Sigma^n, T' \in \Sigma^{n'}} \{C(T')/C(T) : \ed(T, T') = 1\},$$
where $\ed(T, T')$ denotes the edit distance between $T$ and $T'$,
$n' = n$ for substitutions, $n = n+1$ for insertions,
and $n' = n-1$ for deletions.
This is a natural measure for the robustness of compression algorithms
in terms of errors and/or dynamic changes occurring in the input string.
Such errors and dynamic changes are common in real-world scenario including
DNA sequencing and versioned document maintenance.

Following the earlier work of Lagarde and Perifel~\cite{LagardeP18}
and Akagi et al.~\cite{AkagiFI2023},
string compressors and repetitiveness measures can be categorized into three classes:
\begin{description}
  \item[(A) Stable:] Those whose sensitivity is $O(1)$;
  \item[(B) Changeable:] Those whose sensitivity is $\polylog(n)$;
  \item[(C) Catastrophic:] Those whose sensitivity is $O(n^c)$ with some constant $0 < c \leq 1$.
    
\end{description}
For instance, it is shown in~\cite{AkagiFI2023} that
Class (A) includes
the substring complexity~\cite{KociumakaNP23},
the smallest macro scheme~\cite{StorerS82},
the Lempel-Ziv 77 families~\cite{LZ77,StorerS82},
and the smallest grammar~\cite{Rytter03,CharikarLLPPSS05}.
On the other hand, 
Class (B) includes run-length Burrows-Wheeler transform (RLBWT)~\cite{AkagiFI2023,GiulianiILPST21,GiulianiILRSU23},
and the Lempel-Ziv 78~\cite{LZ78} belongs to Class (C)~\cite{LagardeP18}.

The focus of this present article is to analyze the sensitivity of CDAWGs.
In case where the edit operation to the string $T$ is performed at either end of $T$,
then the multiplicative sensitivity of CDAWGs is known to be asymptotically at most $2$, and it is tight~\cite{InenagaHSTAMP05,FujimaruNI25}.
However, the general case with an arbitrary single-character-wise edit on $T$ was not well understood for CDAWGs.
In this paper, we prove that any edit operation at an arbitrary position
on the string can increase the size of the CDAWG asymptotically \emph{at most 8 times larger} than the original, implying that CDAWGs belong to Class (A).
We emphasize that the only known upper bound
for the sensitivity of CDAWGs is $O(n / \log n)$,
which trivially follows since $\size \in O(n)$ and $\size \in \Omega(\log n)$ for any string of length $n$~\cite{Blumer1987,BelazzouguiC17}.
Our technique for proving the constant sensitivity of CDAWGs is purely combinatorial, which involves new and original ideas that were not present in the special case of left/right-end edits~\cite{InenagaHSTAMP05,FujimaruNI25}.
Also, all our arguments hold without a common sentinel symbol $\$$ at the right-end of strings.

Although the CDAWG size $\size$ is regarded as a weak string repetitiveness measure
(as $\size \in \Theta(n)$ for string $\mathtt{a}^{n-1}\mathtt{b}$~\cite{Blumer1987}, and reversing the string can increase $\size$ by a factor of $O(\sqrt{n})$~\cite{InenagaK24} for some string),
our result shows another virtue of CDAWGs being stable in terms of sensitivity and thus being robust against errors and edits.
To our knowledge, CDAWGs are the first compressed indexing structure
proven to achieve $O(1)$ multiplicative sensitivity.

\section{Preliminaries}

\subsection{Strings}

Let $\Sigma$ be an \emph{alphabet} of size $\sigma$.
An element of $\Sigma^*$ is called a \emph{string}.
%The set of characters that occur in a string $T$ is denoted by ${\Sigma_T}$.
For a string $T \in \Sigma^*$, the length of $T$ is denoted by $|T|$.
The \emph{empty string}, denoted by $\varepsilon$, is the string of length $0$.
%Let $\Sigma^+ = \Sigma^* \setminus \{\varepsilon\}$.
For any non-negative integer $n \geq 0$,
let $\Sigma^n$ denote the set of strings of length $n$.
For any two strings $S$ and $T$,
let $\ed(S,T)$ denote the edit distance between $S$ and $T$.
For any string $T$ and a non-negative integer $\ell \geq 0$,
let $\mathcal{K}(T,\ell) = \{S \mid \ed(S,T) = \ell\}$.

For string $T = uvw$, $u$, $v$, and $w$ are called a \emph{prefix}, \emph{substring},
and \emph{suffix} of $T$, respectively.
The sets of prefixes, substrings, and suffixes of string $T$ are denoted by
$\Prefix(T)$, $\Substr(T)$, and $\Suffix(T)$, respectively.
For a string $T$ of length $n$, $T[i]$ denotes the $i$th character of $T$
for $1 \leq i \leq n$,
and $T[i..j] = T[i] \cdots T[j]$ denotes the substring of $T$ that begins at position $i$ and ends at position $j$ on $T$ for $1 \leq i \leq j \leq n$.

\subsection{Maximal substrings and maximal repeats}

A substring $w \in \Substr(T)$ of $T$ is said to be
\emph{left-maximal} if (1) $w \in \Prefix(T)$
or (2) there exist two distinct characters $a,b \in \Sigma$ such that
$aw, bw \in \Substr(T)$,
and it is said to be \emph{right-maximal}
if (1) $w \in \Suffix(T)$ or (2) there exist two distinct characters $a,b \in \Sigma$ such that $wa, wb \in \Substr(T)$.
These substrings $w$ that occur at least twice in $T$
are also called \emph{left-maximal repeats}
and \emph{right-maximal repeats} in $T$, respectively.

Let $\LeftM(T)$ and $\RightM(T)$ denote the sets of left-maximal
and right-maximal substrings in $T$.
Let $\M(T) = \LeftM(T) \cap \RightM(T)$.
The elements in $\M(T)$ are called \emph{maximal substrings} in $T$,
and the elements in $\M(T) \setminus \{T\}$ are called
\emph{maximal repeats} in $T$.
A character $a \in \Sigma$ is said to be 
a \emph{right-extension} of a maximal repeat $w$ of $T$
if $wa \in \Substr(T)$.

%Let $\LeftM(T)$ and $\RightM(T)$ denote the sets of left-maximal
%and right-maximal repeats in $T$.
%Let $\M(T) = \LeftM(T) \cap \RightM(T)$.
%The elements in $\M(T)$ are called \emph{maximal-repeats} in $T$.
%A character $a \in \Sigma$ is said to be 
%a \emph{right-extension} of a maximal repeat $w$ of $T$
%if $wa \in \Substr(T)$.

For any substring $w$ of a string $T$,
we define its \emph{left-representation} and \emph{right-representation}
by $\lrep_T(w) = \alpha w$ and $\rrep_T(w) = w \beta$,
where $\alpha, \beta \in \Sigma^*$ are the shortest strings
such that $\alpha w$ is left-maximal in $T$
and $w\beta$ is right-maximal in $T$, respectively.

\subsection{CDAWGs}

The \emph{compact directed acyclic word graph} (\emph{CDAWG}) of a string $T$,
denoted $\CDAWG(T)$,
is the minimal DFA that recognizes all substrings of $T$,
in which each transition (edge) is labeled by a non-empty substring of $T$.
$\CDAWG(T)$ has a unique source that represents the empty string $\varepsilon$
and a unique sink that represents $T$.
All the other internal nodes represent the maximal repeats in $T$,
namely, the set of the longest strings represented by the nodes of $\CDAWG(T)$ are equal to $\M(T)$.
See Figure~\ref{fig:cdawg} for a concrete example of CDAWGs.

The \emph{size} of $\CDAWG(T)$ for a string $T$ of length $n$
is the number $\size(T)$ of edges in $\CDAWG(T)$,
which is equal to the number of right-extensions of maximal repeats in $T$.

In what follows, we will identify
maximal substrings with CDAWG nodes,
and right-extensions of maximal repeats with CDAWG edges,
respectively.

For each $x \in \M(T)$,
let $\D_T(x)$ denote the number of out-edges of of node $x$ in $\CDAWG(T)$.
It is clear that $\size(T) = \sum_{x \in \M(T)} \D_T(x)$.

\begin{figure}[tbh]
  \centering
  \includegraphics[keepaspectratio,scale=0.35]{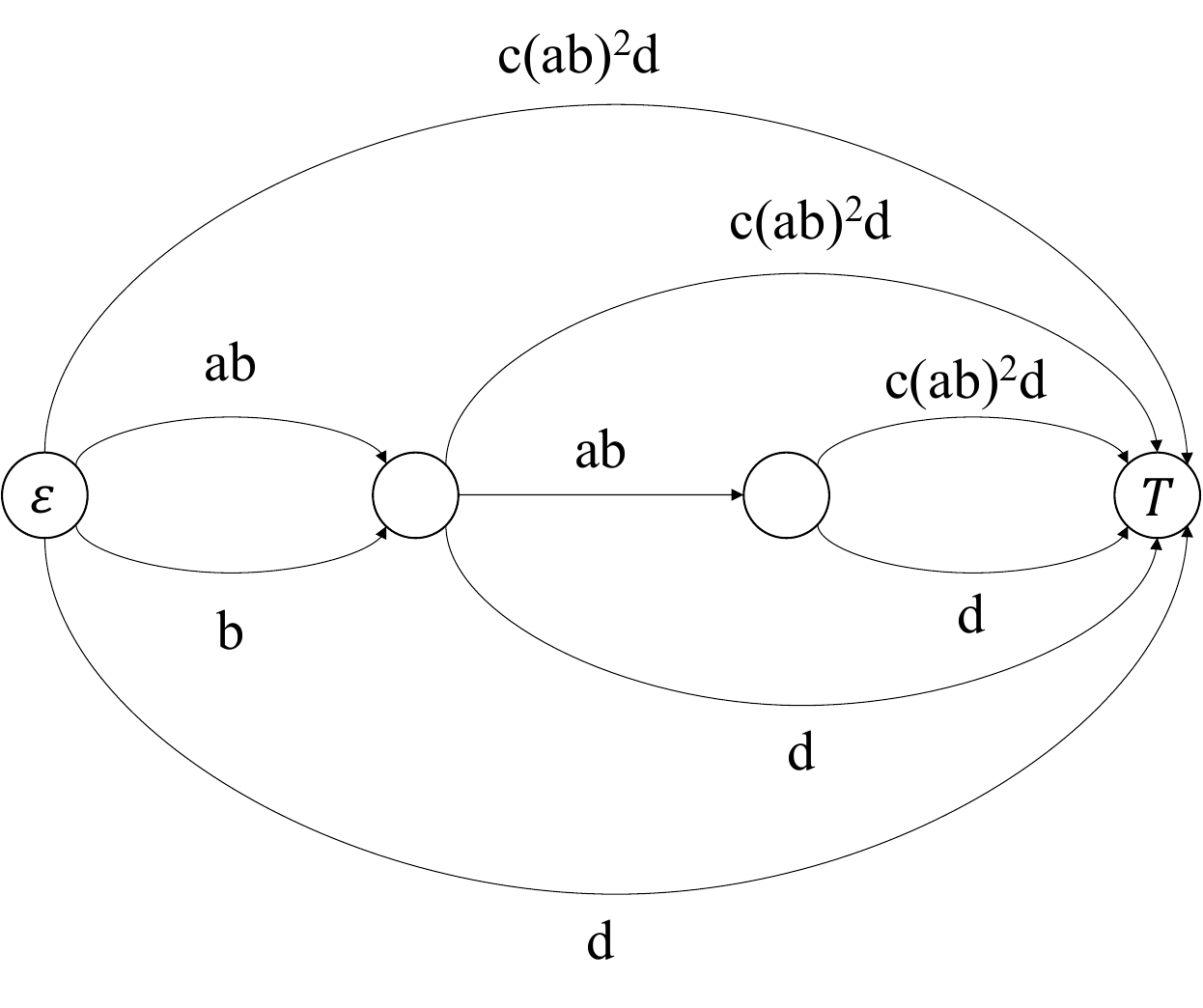}
  \caption{Illustration for $\CDAWG(T)$ of string $T=\mathrm{(ab)^2 c(ab)^2d}$. The longest strings represented by the nodes of $\CDAWG(T)$ are the maximal substrings in $\M(T) = \{\mathrm{\varepsilon, ab, (ab)^2, (ab)^2 c(ab)^2d}\}$.}
  \label{fig:cdawg}
\end{figure}

\subsection{Sensitivity of CDAWG size and our results}

Using the measure $\size$, we define
the worst-case multiplicative \emph{sensitivity} of the CDAWG
with edit operations (resp. insertion, deletion, and substitution) by:
\begin{eqnarray*}
  \MSIns(\mathsf{\size}, n) & = & \max_{T \in \Sigma^n, T' \in \mathcal{K}(T,1) \cap \Sigma^{n+1}} \{\size(T')/ \size(T)\}, \\
  \MSDel(\size, n) & = & \max_{T \in \Sigma^n, T' \in \mathcal{K}(T,1) \cap \Sigma^{n-1}} \{\size(T')/\size(T)\}, \\
  \MSSub(\size, n) & = & \max_{T \in \Sigma^n, T' \in \mathcal{K}(T, 1) \cap \Sigma^{n}} \{\size(T')/\size(T)\}.
\end{eqnarray*}

We will prove the following:
\begin{theorem} \label{theo:CDAWG_sensitivity}
  For any string $T$ of length $n$,
  $\MSIns(\mathsf{\size}, n) \leq (8\size+4)/\size$,
  $\MSDel(\mathsf{\size}, n) \leq (8\size+4)/\size$,
  $\MSSub(\mathsf{\size}, n) \leq (8\size+4)/\size$ hold,
  where $\size = \sum_{x \in \M(T)}\D_T(x)$.
\end{theorem}

Our proof for Theorem~\ref{theo:CDAWG_sensitivity} 
handles all cases of insertions, deletions, and substitutions.

\section{Occurrences of maximal repeats crossing the edited position}

% In what follows, we concentrate on the case of substitutions
% and show that $\size(T') \leq 7\size+3$
% for any strings $T$ and $T'$, where $T'$ is obtained by substituting
% a character in $T$ with a different character.

To present an upper bound for the sensitivity of CDAWG size,
it is essential to consider new occurrences of maximal repeats
that contain or touch the edited position $i$.
In this section, we introduce several new definitions regarding
those new occurrences of maximal repeats.
%and categorize them
%in a way that allows for charging new maximal repeats in $T'$ to existing maximal repeats in $T$ and allows for bounding the number of edges in $\CDAWG(T')$
%just a constant factor away from the number of edges in $\CDAWG(T)$.

For ease of discussions, we sometimes identify an interval where
a substring $w$ occurs in $T'$ with the substring $w$ itself,
when there is no risk of confusion.

% %\sinote*{changed}{%
% \begin{definition}
% Let $x = T'[j..k]$ be a non-empty substring of $T'$ that \emph{touches or contains}
% the edited position $i$, namely
% (1) $k = i-1$ (touching $i$ from left),
% (2) $j \leq i \leq k$ (containing $i$), or
% (3) $j = i+1$ (touching $i$ from right).
% These occurrence of a substring $x$ in $T'$
% are said to be \emph{crossing occurrences} for
% the edited position $i$.
% We will call these occurrences simply as crossing occurrences of $x$.
% See also Figure~\ref{fig:sub} for illustration.
% \end{definition}

% Let $T[i] = \alpha \in \Sigma$ and $T'[i] = \beta \in \Sigma$.
% The leftmost crossing occurrence $x_L$ of $x$ in $T'$ is denoted as
% \[
% x_L =
%     \begin{cases}
%      P_{x_L} & \mbox{if $x_L$ touches $i$ from left}  \\
%      P_{x_L} \beta S_{x_L} & \mbox{if $x_L$ contains $i$} \\
%      S_{x_L} & \mbox{if $x_L$ touches $i$ from right}
%     \end{cases}
% \]
% where $P_{x_L}, S_{x_L} \in \Sigma^*$.
% We denote the leftmost crossing occurrence $x_R$ of $x$ in $T'$ analogously.

\subsection{Crossing occurrences}

\begin{definition}\rhnote*{changed}{%
  Let $x = T'[j..k]$ be a non-empty substring of $T'$ that \emph{touches or contains}
  the edited position $i$.
  That is, if the edit operation is insertion and substitution,
  (1) $k = i-1$ (touching $i$ from left),
  (2) $j \leq i \leq k$ (containing $i$), or
  (3) $j = i+1$ (touching $i$ from right).
  If the edit operation is deletion, 
  (1) $k = i-1$ (touching $i$ from left),
  (2) $j \leq i-1 \land i \leq k$ (containing $i$), or
  (3) $j = i$ (touching $i$ from right).
  }%
  These occurrences of a substring $x$ in $T'$
  are said to be \emph{crossing occurrences} for
  the edited position $i$.
  We will call these occurrences simply as crossing occurrences of $x$.
%  See also Figure~\ref{fig:crossing} for illustration.
  \end{definition}

\rhnote*{added "of x" and changed "the original string T" to "T'"}{%
We denote the left most crossing occurrence $T'[j'..k']$ of $x$ as $x_L$.
For $x_L$, we consider the following substrings $P_{x_L}$ and $S_{x_L}$
of $T'$ (see Figure~\ref{fig:x_Lp_L} for illustration):
}%

In the case that the edit operation is insertion or substitution, let
\begin{eqnarray*}
P_{x_L} & = & 
    \begin{cases}
     T'[j'..i] & \mbox{if $x_L$ touches $i$ from left or contains $i$,}\\
     \varepsilon & \mbox{if $x_L$ touches $i$ from right,}
    \end{cases}\\
S_{x_L} & = &
    \begin{cases}
     \varepsilon & \mbox{if $x_L$ touches $i$ from left,}  \\
     T'[i..k'] & \mbox{if $x_L$ contains $i$ or touches $i$ from right.}
    \end{cases}
\end{eqnarray*}
In the case that the edit operation is deletion, let
\begin{eqnarray*}
P_{x_L} & = &
    \begin{cases}
     T'[j'..i-1] & \mbox{if $x_L$ touches $i$ from left or contains $i$,}\\
     \varepsilon & \mbox{if $x_L$ touches $i$ from right,}
    \end{cases}\\
S_{x_L} & = &
    \begin{cases}
     \varepsilon & \hspace*{6mm} \mbox{if $x_L$ touches $i$ from left,}  \\
     T'[i..k'] & \hspace*{6mm} \mbox{if $x_L$ contains $i$ or touches $i$ from right.}
    \end{cases}
\end{eqnarray*}
We define the rightmost crossing occurrence $x_R$, together with $P_{x_R}$ and $S_{x_R}$, analogously.

%}%
%\begin{definition}
%  Let $T'$ be the string obtained from $T$ by substituting the $i$th character $T[i] = \alpha$ with $\beta~(\beta \neq \alpha)$.
%  Assume that $x=P \beta S \: (P,S \in {\Sigma_x}^*, \beta \in \Sigma_x)$.
%  Then we define crossing occurrences of string as follows.
%  The occurrence of $x$ in $T'$ beginning at position $i-|P|$ is called a \emph{crossing occurrence} of $x$ in $T'$
%  (See Figure~\ref{fig:sub} for an illustration).
%  The leftmost crossing occurrence of $x$ in $T'$ is denoted by
%  $x_L = P_{x_L} \beta S_{x_L}$.
%  Similarly, the rightmost crossing occurrence of $x$ in $T'$ is denoted by $x_R = P_{x_R} \beta S_{x_R}$.
%\end{definition}

%\begin{figure}[H]
%  \centering
%  \includegraphics[keepaspectratio,scale=0.35]{crossing.pdf}
%  \caption{Illustration for crossing occurrences of $x$ in $T'$.}
%  \label{fig:crossing}
%\end{figure}

\begin{figure}[tbh]
  \centering
  \includegraphics[keepaspectratio,scale=0.35]{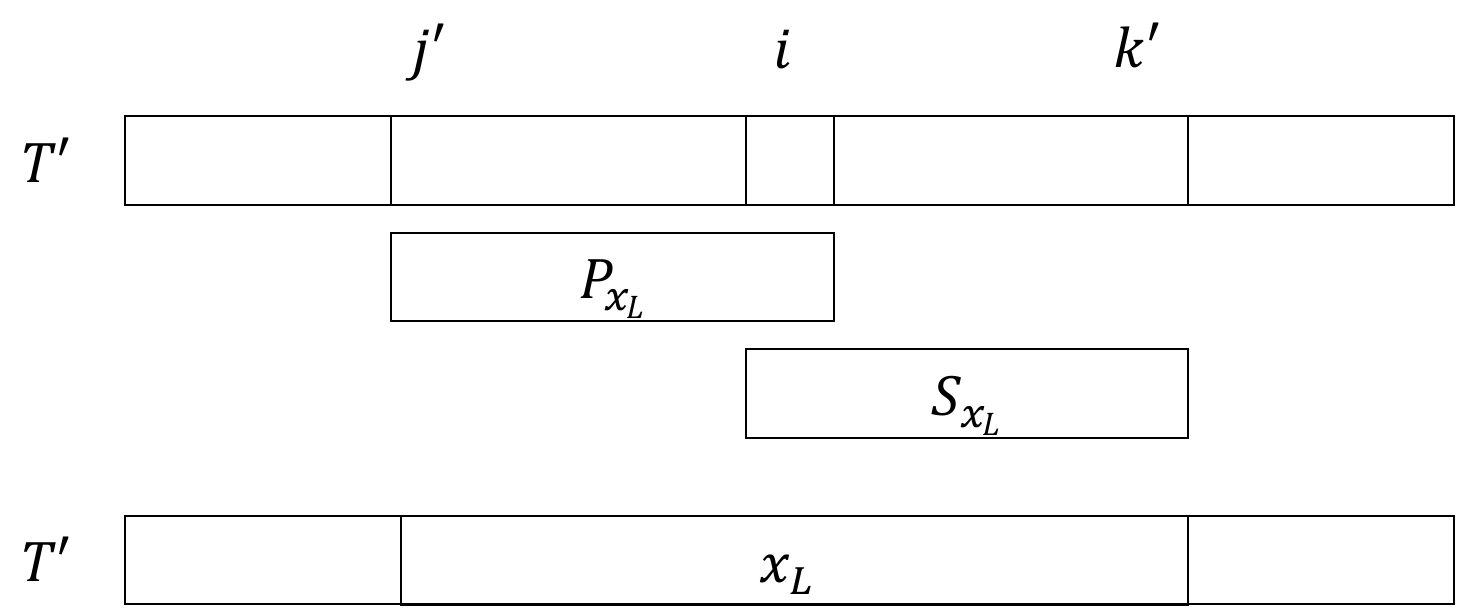}
  \caption{Illustration of $x_L$ in $T'$ for the case where $x_L$ contains $i$, with insertion and substitution.}
  \label{fig:x_Lp_L}
\end{figure}

\begin{definition}
  We categorize strings $x$ that have crossing occurrence(s) in the edited string in $T'$ into the five following types, depending on the properties of $x$:
  \begin{description}
    \item {Type (i)}: $x$ has only one crossing occurrence of $x$ in $T'$.
    \item {Type (ii)}: 
      \begin{enumerate}
        \item $x$ has two or more crossing occurrences of $x$ in $T'$.
        \item If all occurrences of $x$ in $T'$ are crossing occurrences of $x$ in $T'$, then $x \notin \LeftM(T')$ and $x \notin \RightM(T')$.
      \end{enumerate}
    \item {Type (iii)}:
      \begin{enumerate}
        \item $x$ has two or more crossing occurrences of $x$ in $T'$.
        \item If all occurrences of $x$ in $T'$ are crossing occurrences of $x$ in $T'$, then $x \notin \LeftM(T')$ and $x \in \RightM(T')$.
      \end{enumerate}
    \item {Type (iv)}:
      \begin{enumerate}
        \item $x$ has two or more crossing occurrences of $x$ in $T'$.
        \item If all occurrences of $x$ in $T'$ are crossing occurrences of $x$ in $T'$, then $x \in \LeftM(T')$ and $x \notin \RightM(T')$.
      \end{enumerate}
    \item {Type (v)}
      \begin{enumerate}
        \item $x$ has two or more crossing occurrences of $x$ in $T'$.
        \item If all occurrences of $x$ in $T'$ are crossing occurrences of $x$ in $T'$, then $x \in \M(T')$.
      \end{enumerate}
  \end{description}  
\end{definition}

Figure~\ref{fig:case} illustrates the aforementioned five types of a string $x$ when $x_L \notin \Prefix(T')$ and $x_R \notin \Suffix(T')$.

\begin{figure}[H]
  \centering
  \includegraphics[keepaspectratio,scale=0.35]{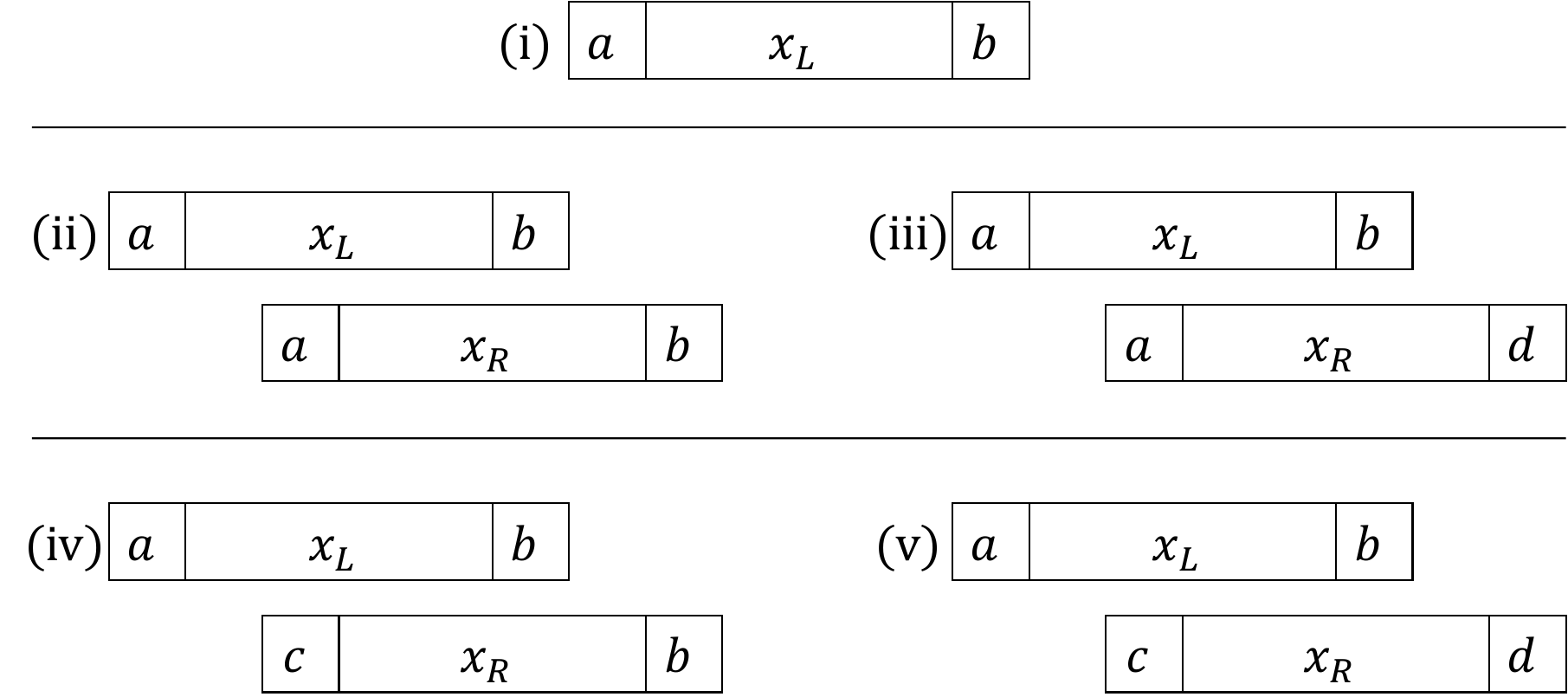}
  \caption{Illustration for the five cases of a string $x$ when $x_L \notin \Prefix(T')$ and $x_R \notin \Suffix(T')$, where $a \ne c$ and $b \ne d$ for characters $a,b,c,d \in \Sigma$.}
  \label{fig:case}
\end{figure}

The reason why we only consider $x_L$ and $x_R$ is due to the periodicity for $x$.
We remark that when $x$ have multiple crossing occurrences which contain or touch the edited position $i$ in $T'$,
then the characters immediately before all crossing occurrences of $x$ in $T'$ except for $x_L$ are the same and the characters immediately after all crossing occurrences of $x$ in $T'$ except for $x_R$ are the same.
Therefore, we do not need to consider all crossing occurrences of $x$ in $T'$ except for $x_L$ and $x_R$ to examine whether $x$ is maximal in $\M(T')$ or not.

\subsection{New/Existing maximal repeats}

In Definitions~\ref{def:new_maximal_repeats} and~\ref{def:existing_maximal_repeats} below, we introduce $\N$ and $\Q$ which are respectively
the sets of new maximal repeats and existing maximal repeats in $T'$,
such that $\N \cup \Q = \M(T') \setminus \{T'\}$.
We then partition each of them into smaller subsets that are suitable for our needs.

\begin{definition} \label{def:new_maximal_repeats}
  Let
  \sinote*{changed}{%
    $\N= (\M(T')\setminus \M(T)) \setminus \{T'\}$
  }%
%  $\N=\M(T')\setminus \M(T)$
  denote the set of new maximal repeats in $T'$.
  We divide $\N$ into the three following disjoint subsets:
%  $\None = \N\cap \RightM(T)$,
%  $\Ntwo = \N\cap \LeftM(T)$, and
%  $\Nthree = \N \setminus (\None \cup \Ntwo)$.
%  Let the subsets $\None$, $\Ntwo$ and $\Nthree$ of $\N$ be defined as
    \begin{eqnarray*}
    \None & = & \N\cap \RightM(T), \\
    \Ntwo & = & \N\cap \LeftM(T), \\
    \Nthree & = & \N \setminus (\None \cup \Ntwo).
    \end{eqnarray*}
  Further, we divide $\Nthree$ into two disjoint subsets $\Nv$ and $\Nnotv$ as follows:
  $\Nv$ is the set of strings $x \in \Nthree$ such that
  (1) $x$ is of Type $\rm{(v)}$, and
  (2) there is no other right-extension of $x$ in $T'$ than the right-extension(s) of the crossing occurrence(s) of $x$,
  and $\Nnotv = \Nthree \setminus \Nv$.
%    let $\Nv$ denote a set of $x \in \Nthree$ satisfying the two following conditions:
%    (1) $x$ is of Type $\rm{(v)}$; 
%    (2) There is no distinct right-extension of $x$ in $T'$ other than the right-extension(s) of the crossing occurrence(s) of $x$.
%    Let $\Nnotv = \Nthree \setminus \Nv$.

%  Further, let $\Nv$ denote the subset of $\Nthree$
%  of which strings is not substring of $T$.
%  Let $\Nnotv = \Nthree \setminus \Nv$, which is the subset of strings contained within $T$.
%
  %  In addition, define a set of strings that are element of $\Nthree$ and case (v) as $\Nv$, 
%  also, define a set of strings that are element of $\Nthree$ and not case (v) as $\Nnotv$.
  %
\end{definition}

We note that the following properties hold for $\None,\Ntwo$ and $\Nthree$ by definition:
\begin{itemize}
\item $x\in \None\Rightarrow x\notin \LeftM(T)$ because $x \notin \M(T)$.
\item $x\in \Ntwo\Rightarrow x\notin \RightM(T)$ because $x \notin \M(T)$.
\item $x\in \Nthree\Rightarrow x\notin \RightM(T) \land x\notin \LeftM(T)$.
\end{itemize}

\begin{definition} \label{def:existing_maximal_repeats}
  Let
  \sinote*{changed}{%
    $\Q= (\M(T') \cap \M(T)) \setminus \{T'\}$
  }%
  %  $\Q=\M(T')\setminus \N = \M(T') \cap \M(T)$
  denote
  the set of existing maximal repeats in $T'$.
  We divide $\Q$ into the two following subsets:
  \begin{eqnarray*}
    \Qn & = & \{x \in \Q \mid \D_{T'}(x) > \D_T(x)\},\\
    \Qnotn & = & \{x \in \Q \mid \D_{T'}(x) \leq \D_T(x)\}.
  \end{eqnarray*}
%  Let $\Qn$ be the subset of $Q$ whose string's right-extensions increase.
%  Let $\Qnotn$ be the subset of $Q$ whose string's right-extensions do not increase.
\end{definition}
Namely, $\Qn$ (resp. $\Qnotn$) is the set of existing nodes of the CDAWG
for which the number of out-edges increase (resp. do not increase). 

\begin{example}
  \label{ex:mr}
  Consider the string $T = \mathrm{cabcabcdabca}\mathbf{d}\mathrm{bcabcdabcabdcabcabcabdabcab}$
  and \\
  the string $T' = \mathrm{cabcabcdabca|bcabcdabcabdcabcabcabdabcab}$ obtained by deleting the character $T[13] = \mathbf{d}$ from $T$ highlighted in bold.
  The edited position in $T'$ is designated by a~$|$.
  For instance, 
  $\mathrm{dabcab}\in\Qn, \mathrm{bcabc}\in\Qnotn, \mathrm{abcabc}\in\None, \mathrm{abcabcab}\in\Nnotv, \mathrm{cabcabcdabcab}\in\Nv$
  (also see Figure~\ref{fig:ex_mr}).
\end{example}

\begin{figure}[H]
  \centering
  \includegraphics[keepaspectratio,scale=0.3]{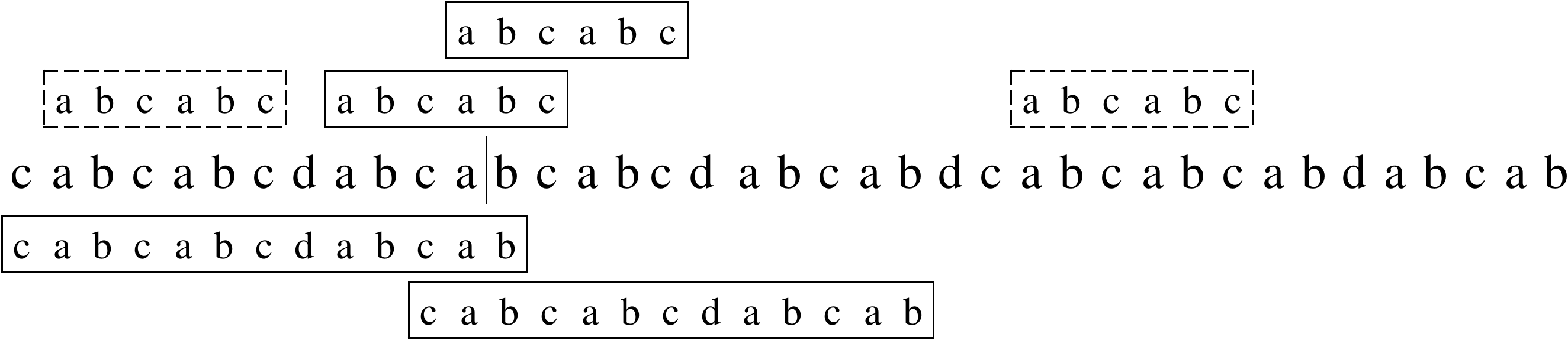}
  \caption{Illustration for $T'= \mathrm{cabcabcdabca|bcabcdabcabdcabcabcabdabcab}$ in Example~\ref{ex:mr} and the occurrences of $\mathrm{abcabc}\in\None$ and $\mathrm{cabcabcdabcab}\in\Nv$ in $T'$.
  The~$|$ symbol in $T'$ exhibits the edit position. The solid line boxes exhibit the crossing occurrences of $\mathrm{abcabc}$ and $\mathrm{cabcabcdabcab}$ in $T'$,
  and the dashed line boxes exhibit the non-crossing occurrences of them in $T'$.}
  \label{fig:ex_mr}
\end{figure}

Recall that $\size = \sum_{x \in \M(T)}\D_T(x)$ denotes the number of edges
in $\CDAWG(T)$ before the edit.
In the subsequent sections,
we work on the three disjoint subsets
$\None \cup \Nnotv$, 
$\Ntwo \cup \Q$, and
$\Nv$ of $\M(T')$,
and show that $\sum_{x \in \None \cup \Nnotv}\D_{T'}(x) \le 3\size+2$ (Section~\ref{sec:upper_bound_1}),
$\sum_{x \in \Ntwo \cup \Q}\D_{T'}(x) \le 3\size+2$ (Section~\ref{sec:upper_bound_2}),
and $\sum_{x \in \Nv}\D_{T'}(x) \le 2\size$ (Section~\ref{sec:upper_bound_3}).
All these immediately lead to Theorem~\ref{theo:CDAWG_sensitivity}
that upper bounds the number of edges in $\CDAWG(T')$ after the edit to $8\size + 4$.

%  By Lemma~\ref{lem:dt1}, Lemma~\ref{lem:dt2} and Lemma~\ref{lem:dt3},
%  we obtain Theorem~\ref{theo:CDAWG_sensitivity}.

\section{Upper bound for total out-degrees of nodes w.r.t. $\None \cup \Nnotv$}
\label{sec:upper_bound_1}

  In this section, we show an upper bound for the total out-degrees of the nodes corresponding to strings in $\None \cup \Nnotv  \subseteq \M(T')$.
  Recall that $x \in \None \cup \Nnotv$ implies $x \notin \LeftM(T)$.
  
We first describe useful properties of
strings $x \in \None \cup \Nnotv$.
  
  \begin{lemma} \label{lem:exist1}
    Any $x \in \None \cup \Nnotv$ occurs in $T$.
  \end{lemma}

  \begin{proof}
%  We prove Lemma~\ref{lem:exist1} by exhaustion.
    In the case $x \in \None$, since $x \in \RightM(T)$, 
    $x$ occurs in $T$. 

    Let us consider the case $x \in \Nnotv$ that is of Type $\rm{(i)}$, $\rm{(ii)}$, $\rm{(iii)}$ or $\rm{(iv)}$.
    Since $x$ is not of Type $\rm{(v)}$, 
    if all occurrences of $x$ in $T'$ are crossing occurrences of $x$ in $T'$, then $x \notin \M(T')$.
    Therefore, $x$ occurs in $T$.

    Let us consider the case $x \in \Nnotv$ that is of Type $\rm{(v)}$.
    Due to the definition of $\Nnotv$, there exists a distinct right-extension of $x$ in $T'$ other than the right-extension(s) of the crossing occurrence(s) of $x$.
    Therefore, there is a non-crossing occurrence of $x$ in $T'$
    as shown in Figure~\ref{fig:exsit1},
    implying that $x$ occurs in $T$.
  \end{proof}

  \begin{figure}[hbt]
    \centering
    \includegraphics[keepaspectratio,scale=0.33]{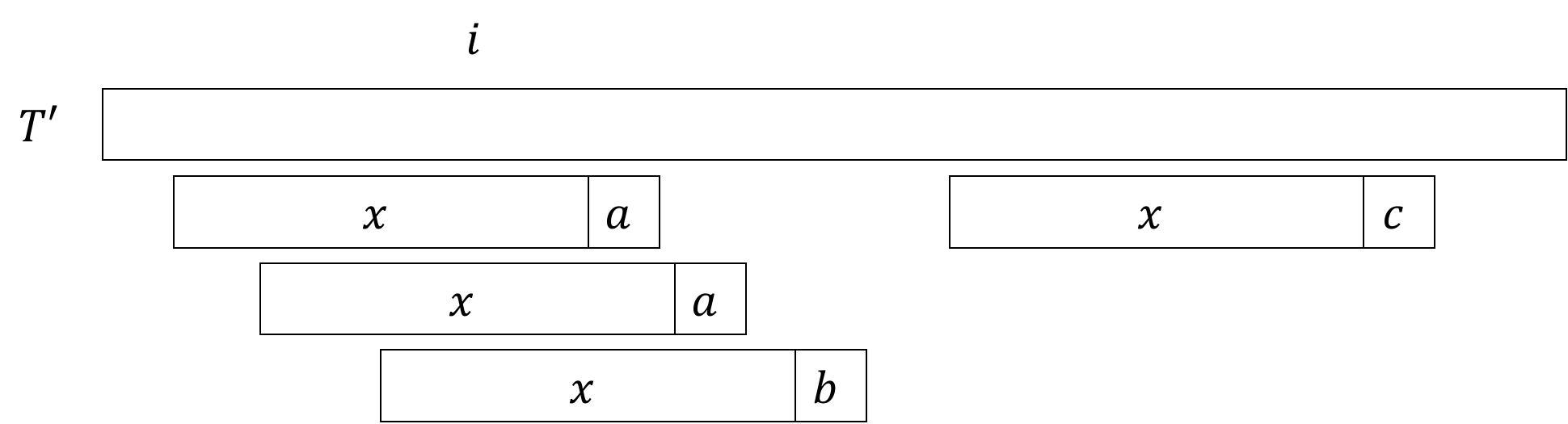}
    \caption{Illustration for Lemma~\ref{lem:exist1} where $i$ is the edited position and $a, b, c$ differ from each other.}
    \label{fig:exsit1}
  \end{figure}

%  \subsection{Case that $x\in \None \cup \Nnotv$ contains the edit position}

  \begin{lemma} \label{lem:sp123}
    For any $x \in \None \cup \Nnotv$ that is of Type $\rm{(i)}$, $\rm{(ii)}$ or $\rm{(iii)}$,
    there does not exist $y\in \None \cup \Nnotv$ such that
    $|y|>|x|$ and $S_{x_L}=S_{{y}_{G}}$,
    where $G \in \{L,R\}$.
  \end{lemma}
  
  \begin{proof}
    If $x_L$ is a prefix of $T'$, then clearly there is no $y$ satisfying
    $|y|>|x|$ and $S_{x_L}=S_{{y}_{G}}$.
    In what follows, we consider the case that $x_L$ is not a prefix of $T'$.
    
    For a contrary, suppose that for $x \in \None \cup \Nnotv$ that is of Type $\rm{(i)}$, $\rm{(ii)}$ or $\rm{(iii)}$, 
    there exists $y\in \None \cup \Nnotv$ such that
    $|y|>|x|$ and $S_{x_L}=S_{{y}_{G}}$, where $G \in \{L,R\}$.
    See also Figure~\ref{fig:sp123}.
    Let $a$ be the character immediately before $x_L$.
    Since $x$ is of Type $\rm{(i)}$, $\rm{(ii)}$ or $\rm{(iii)}$,
    every crossing occurrence of $x$ in $T'$ is immediately
    preceded by $a$.
    Because $x \in \M(T')$,
    it holds that $x \in \Prefix(T')$,
    or there is a distinct character $b \in \Sigma \setminus \{a\}$
    such that $bx$ occurs in $T'$.
    This implies that there is a non-crossing occurrence of $x$ in $T'$,
    which is as a prefix of $T$ or is immediately preceded by $b$ in $T$.
    By Lemma~\ref{lem:exist1}, $y$ occurs in $T$, and thus $ax$ that is a suffix of $y$ also occurs in $T$.
    Hence $x \in \LeftM(T)$, however, this contradicts that $x \notin \LeftM(T)$.
  \end{proof}

%  Lemma~\ref{lem:sp123} states that $x$ and $y$ cannot occur as in Figure~\ref{fig:sp123}

  \begin{figure}[bth]
    \centering
    \includegraphics[keepaspectratio,scale=0.33]{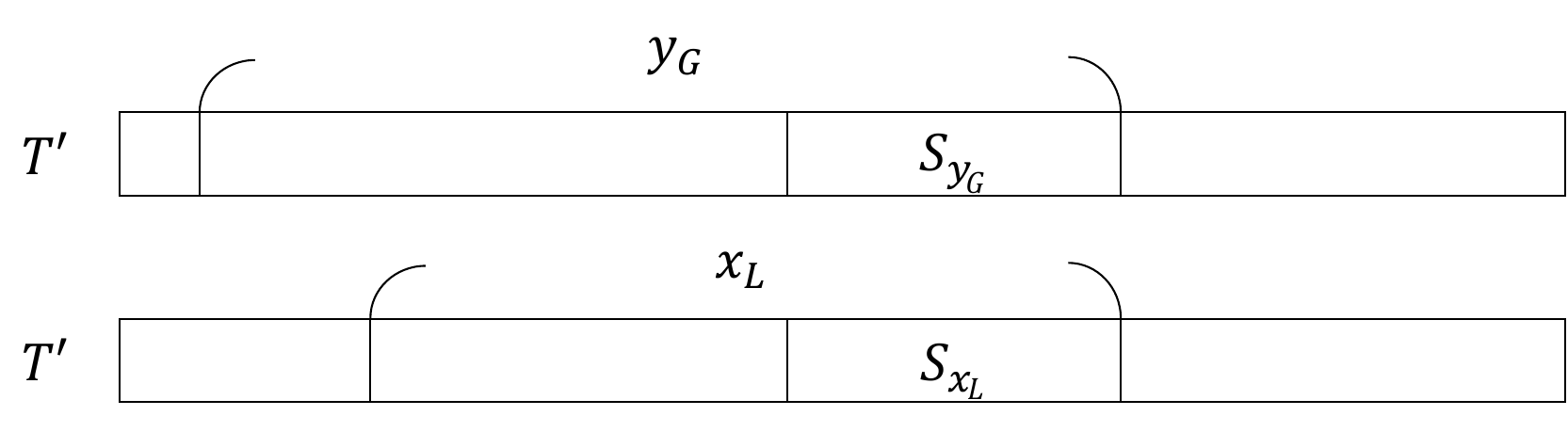}
    \caption{Illustration for Lemma~\ref{lem:sp123}: impossible occurrences of $x$ and $y$ with $S_{x_L} = S_{y_G}$.}
    \label{fig:sp123}
  \end{figure}

  \begin{lemma} \label{lem:sp45}
    For any $x \in \None \cup \Nnotv$ that is of Type $\rm{(iv)}$ or $\rm{(v)}$, 
    there do not exist $y,z \in \None \cup \Nnotv$ with
    $|y|>|x|$ and $|z|>|x|$
    satisfying
    $S_{x_L}=S_{{y}_{G}}$ and $S_{x_R}=S_{{z}_{F}}$ simultaneously,
    where $G,F \in \{L,R\}$.
  \end{lemma}

  \begin{proof}
    If $x_L$ is a prefix of $T'$, then clearly there is no $y$ satisfying
    $|y|>|x|$ and $S_{x_L}=S_{{y}_{G}}$.
    In what follows, we consider the case that $x_L$ is not a prefix of $T'$.
    
    For a contrary, 
    suppose that for $x \in \None \cup \Nnotv$ that is of Type $\rm{(iv)}$ or $\rm{(v)}$, 
    there exist $y,z \in \None \cup \Nnotv$ with $|y|>|x|$ and $|z|>|x|$ such that $S_{x_L}=S_{{y}_{G}}$ and $S_{x_R}=S_{{z}_{F}}$ at the same time, where $G,F \in \{L,R\}$.
    Let the character immediately before $x_L$ and the character immediately before $x_R$ be $a$ and $c$~($a \neq c$), respectively.
    \rhnote*{changed "b" to "c"}{%
    By Lemma~\ref{lem:exist1}, $y$ and $z$ occur in $T$, and thus $ax$ that is a suffix of $y$ and $cx$ that is a suffix of $z$ both occur in $T$.
    }%
    Therefore, $x \in \LeftM(T)$, however, this contradicts that $x \notin \LeftM(T)$.
  \end{proof}

  \subsection{Correspondence between $\None \cup \Nnotv$ and $\M(T)$}

  For any $x \in \None \cup \Nnotv$ that is of Type $\rm{(i)}$, $\rm{(ii)}$ or $\rm{(iii)}$, we associate $x$ with $S_{x_L}$.
  For any $x \in \None \cup \Nnotv$ that is of Type $\rm{(iv)}$ or $\rm{(v)}$, 
  if there does not exist $y\in \None \cup \Nnotv$
  such that $|y|>|x|$ and 
  $S_{x_L}=S_{{y}_{G}}$ with $G \in \{L,R\}$, we associate $x$ with $S_{x_L}$, 
  and otherwise we associate $x$ with $S_{x_R}$.

  By Lemma~\ref{lem:sp123} and Lemma~\ref{lem:sp45}, each $x \in \None \cup \Nnotv$ can be associated to a distinct string $S_{x_G}$ with $G \in \{L,R\}$.
  Note however that $S_{x_G}$ may not be maximal in $T$.
  Thus we introduce a function $U$
  that bridges each $x \in \None \cup \Nnotv$ to a distinct maximal substring in $T$.
%  For any $x \in \None \cup \Nnotv$
%  we define $U(x)$ to which $x$ corresponds, by using $S_{x_G} \: (G \in \{L,R\})$ which $x$ corresponds,
%  as follows:

  \begin{definition} \label{def:U_x}
    For any $x \in \None \cup \Nnotv$,
    let $U(x)=\lrep_T({S_{x_G}})$ (see Figure~\ref{fig:U_x}).
  \end{definition}
    By Lemma~\ref{lem:exist1}, $x$ occurs in $T$ and thus
    its suffix $S_{x_G}$ also occurs in $T$.
    Hence $U(x)=\lrep_T({S_{x_G}})$ is well defined.

%  \sinote*{added}{%
%  Recall that $x$ is a maximal repeat in $T'$
%  and thus $x$ occurs at least twice in $T'$,
%  implying that $S_{x_G}$ occurs at least once in $T$.
%  Thus $U(x)=\lrep_T({S_{x_G}})$ is well defined.
%  }%

%  \sinote*{modified}{%
%  \begin{definition} \label{def:U_x}
%    For any $x \in \None \cup \Nnotv$,
%    let
%    \[
%    U(x) =
%    \begin{cases}
%      \lrep_T({S_{x_G}}) & \mbox{for deletion} \\
%      \lrep_T({T[i]S_{x_G}[2..|S_{x_G}|]}) & \mbox{for insertion and substitution}
%    \end{cases}
%    \]
%  \end{definition}
%  }%

%  See Figure~\ref{fig:U_x} for an illustration for $U(x)$.

  \begin{figure}[H]
    \centering
    \includegraphics[keepaspectratio,scale=0.33]{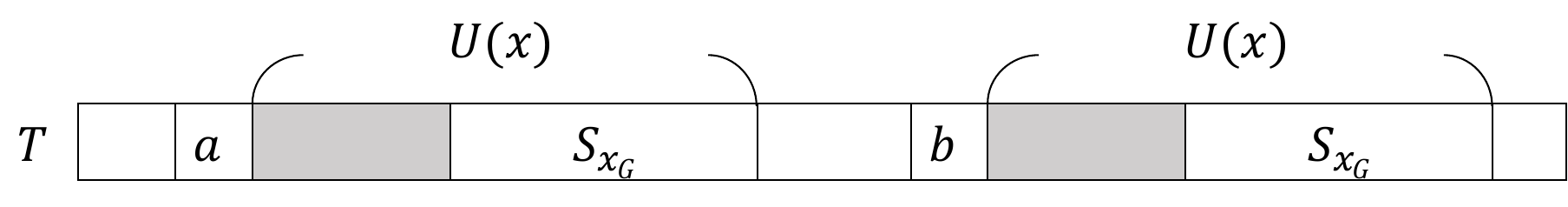}
    \caption{Illustration for $U(x)$ $(a\neq b)$.}
    \label{fig:U_x}
  \end{figure}

  \begin{lemma} \label{lem:U_x}
    For any $x \in \None \cup \Nnotv$, $U(x) \in \M(T)$.
  \end{lemma}

  \begin{proof}
    By Definition~\ref{def:U_x}, $U(x) = \lrep_T({S_{x_G}}) \in \LeftM(T)$.
    Therefore, it suffices for us to prove $U(x) \in \RightM(T)$.
    From now on, we consider the four following cases:
%    \begin{enumerate}
%    \item[(a)] $x \in \None$. 
%    \item[(b)] $x \in \Nnotv$ and $x$ is of Type $\rm{(i)}$, $\rm{(ii)}$ or $\rm{(iv)}$.
%    \item[(c)] $x \in \Nnotv$ and $x$ is of Type $\rm{(iii)}$.
%    \item[(d)] $x \in \Nnotv$ and $x$ is of Type $\rm{(v)}$.
%    \end{enumerate}

    \noindent \textbf{Case (a) $x \in \None$:}
    In this case, $x \in \RightM(T)$, therefore $S_{x_{G}}$ that is a suffix of $x$ also satisfies $S_{x_{G}} \in \RightM(T)$.
    Hence, $U(x) = \lrep_T({S_{x_G}}) \in \RightM(T)$.

    \noindent \textbf{Case (b) $x \in \Nnotv$ and $x$ is of Type $\rm{(i)}$, $\rm{(ii)}$ or $\rm{(iv)}$:}
%    Let us consider the case that $x \in \Nnotv$ and $x$ is in $\rm{(i)}$, $\rm{(ii)}$ or $\rm{(iv)}$.
    Let the character immediately after all crossing occurrences of $x$ in $T'$ be $a$.
    There exists $xb \: (b \ne a)$ in $T'$ or $x \in \Suffix(T')$ because $x \in \M(T')$.
    Since the character immediately after all crossing occurrences of $x$ in $T'$ is $a$, 
    then there exists $xb \: (b \ne a)$ or $x \in \Suffix(T)$ in $T$.
    Hence, $S_{x_{G}} \in \RightM(T)$ since the character immediately after $S_{x_{G}}$ is $a$.
    Thus, $U(x) = \lrep_T({S_{x_G}}) \in \RightM(T)$.
  % $x \in \Suffix(T)$だと$x \in \RightM(T)$となるから考慮しなくていいけど論理的におかしなことは書いてない

   \noindent \textbf{Case (c) $x \in \Nnotv$ and $x$ is of Type $\rm{(iii)}$:}
%    Let us consider the case that $x \in \Nnotv$ and $x$ is in $\rm{(iii)}$.
    Since $x$ is of Type $\rm{(iii)}$, we associate $x$ with $S_{x_L}$.
    Because $x$ is of Type $\rm{(iii)}$ and $S_{x_L}$ is a suffix of a $S_{x_ R}$, $S_{x_{G}} \in \RightM(T)$ holds.
    Thus, $U(x) = \lrep_T({S_{x_G}}) \in \RightM(T)$.

   \noindent \textbf{Case (d) $x \in \Nnotv$ and $x$ is of Type $\rm{(v)}$:}
%    Let us consider the case that $x \in \Nnotv$ and $x$ is in $\rm{(v)}$.
    Let the character immediately after $S_{x_ G}$ be $a$.
    Since there exists a distinct right-extension of $x$ in $T'$ other than the right-extension(s) of the crossing occurrence(s) of $x$, there exists $xb$ $(b \neq a)$ in $T$.
    Therefore, $S_{x_ G} \in \RightM(T)$.
    Thus, $U(x) = \lrep_T({S_{x_G}}) \in \RightM(T)$.

    Consequently, we have $U(x) \in \M(T)$.
  \end{proof}

  The next lemma states the uniqueness of $U(x)$.
  \begin{lemma} \label{lem:U_xU_y}
    For any $x,y\in \None \cup \Nnotv$ with $x \neq y$, $U(x) \neq U(y)$.
  \end{lemma}

  \begin{proof}
    Suppose that there exist $x,y\in \None \cup \Nnotv$ such that
    $x \neq y$ and $U(x) = U(y)$.
    Let $x$ and $y$ correspond to $S_{x_{G}}$ and $S_{y_{F}}$, respectively,
    where $G,F \in \{L,R\}$.
    Let $U(x)=AS_{x_{G}},U(y)=BS_{y_{F}} \:(A,B\in \Substr(T))$,
    and assume without loss of generality that $|S_{x_{G}}|<|S_{y_{F}}|$.
    Then $|A|>|B|$ because $U(x) = U(y)$.
    \rhnote*{added "$U(y)=BS_{y_{F}} \in \M(T)$ by Lemma~\ref{lem:U_x}"}{%
    Since $U(y)=BS_{y_{F}} \in \M(T)$ by Lemma~\ref{lem:U_x}, and since $BS_{x_{G}}$ is a prefix of $BS_{y_{F}}$ (see Figure~\ref{fig:U_xU_y}), we have $BS_{x_{G}}\in \LeftM(T)$.
    }% 
    This contradicts $\lrep_T({S_{x_{G}}})=AS_{x_{G}}$.
  \end{proof}

  \begin{figure}[H]
    \centering
    \includegraphics[keepaspectratio,scale=0.33]{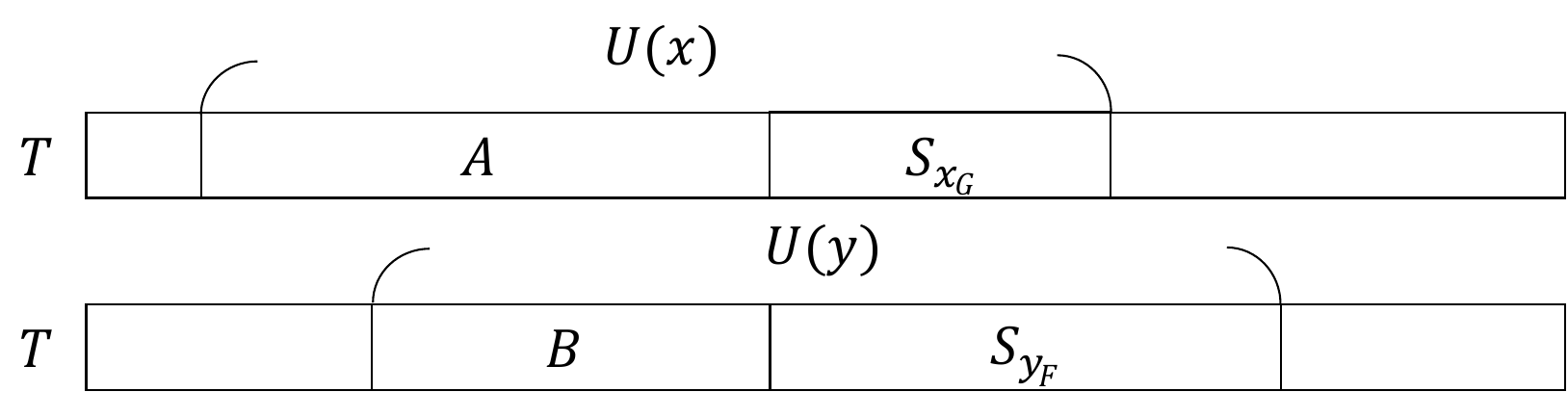}
    \caption{Illustration for the proof of Lemma~\ref{lem:U_xU_y}, where $U(x) = U(y)$.}
    \label{fig:U_xU_y}
  \end{figure}

  \begin{comment}
  \subsection{Case that $x\in \None \cup \Nnotv$ does not contain the edit position}

  From now on, in this subsection, we consider the case that $x\in \None \cup \Nnotv$ does not contain the edit position.
  
  \begin{lemma}
    \label{lem:eps1}
    The number of $x\in \None \cup \Nnotv$ that does not contain the edit position is 1.
  \end{lemma}

  \begin{proof}
    It can be proved by making the same arguments as for Lemma~\ref{lem:sp123} and Lemma~\ref{lem:sp45}.
  \end{proof}%
  \end{comment}

  \subsection{Upper bound w.r.t. $\None \cup \Nnotv$}

  \begin{lemma} 
    \label{lem:dt1}
    $\sum_{x \in \None \cup \Nnotv}\D_{T'}(x) \le 3\size+2$.
  \end{lemma}

  \begin{proof}\rhnote*{changed}{%
    Let $U(x)=\lrep_T({S_{x_G}})$, where $G \in \{L,R\}$.
    Since $S_{x_G}$ is a suffix of $x$, $\D_{T}(x) \leq \D_{T}(U(x))$.
    Since there are at most two distinct characters immediately after the crossing occurrences of $x$, 
    $\D_{T'}(x) \leq \D_{T}(U(x))+2$.
    For $U(x) \neq T$, we have $\D_T(U(x)) \ge 1$. Thus $\D_{T'}(x) \le \D_T(U(x))+2 \le 3\D_T(U(x))$.
    For $U(x) = T$, we have $\D_T(U(x))=0$. Thus $\D_{T'}(x) \le 2$.
    By using Lemma~\ref{lem:U_xU_y} and summing up these, we get 
    $\sum_{x \in \None \cup \Nnotv}\D_{T'}(x) \le {\sum_{x \in \None \cup \Nnotv} 3\D_T(U(x))}+2 \le 3\size+2$.
  \end{proof}
  }%

\section{Upper bound for total out-degrees of nodes w.r.t. $\Ntwo \cup \Q$}
\label{sec:upper_bound_2}

In this section, we show an upper bound for the total out-degrees of nodes corresponding to strings that are elements of $\Ntwo \cup \Q \subseteq \M(T')$.

We first present properties of the strings in $\Ntwo \cup \Q$.
In particular, we focus on the strings in $\Ntwo \cup \Qn$,
as the strings in $\Qnotn$ are less important and can be handled in a trivial manner.

  \begin{lemma} \label{lem:exist2}
    Any $x \in \Ntwo \cup \Qn$ occurs in $T$.
  \end{lemma}

  \begin{proof}
    Since $x \in \Ntwo \cup \Qn$, $x \in \LeftM(T)$. Thus $x \in \Ntwo \cup \Qn$ occurs in $T$.
  \end{proof}

  \begin{lemma} \label{lem:sp124}
    For any $x \in \Ntwo \cup \Qn$ that is of Type $\rm{(i)}$, $\rm{(ii)}$ or $\rm{(iv)}$,
    there does not exist $y\in \Ntwo \cup \Qn$ such that $|y|>|x|$ and $P_{x_{G}}=P_{y_{F}}$, where $G,F \in \{L,R\}$.
  \end{lemma}

  \begin{proof}
    The case that $x \in \Ntwo$ follows from a symmetrical argument to Lemma~\ref{lem:sp123}, in which $y$ may belong to $\Ntwo$ or $\Qn$.
    \rhnote*{delete $\Qn = \M(T) \cap \M(T')$}{%
    Let us consider the case that $x \in \Qn$.
    }%
    Suppose that for $x \in \Qn$ which is of Type $\rm{(i)}$, $\rm{(ii)}$ or $\rm{(iv)}$,
    there is $y\in \Ntwo \cup \Qn$ such that $|y|>|x|$ and $P_{x_{G}}=P_{y_{F}}$, where $G,F \in \{L,R\}$.
    If $x_R$ is a suffix of $T'$, then there is no $y$ such that $|y|>|x|$ and $P_{x_{G}}=P_{y_{F}}$.
    From now on consider the case that $x_R$ is not a suffix of $T'$.
    Let $b$ be the character immediately after $x_{G}$.
    Then, since $x$ is of Type $\rm{(i)}$, $\rm{(ii)}$ or $\rm{(iv)}$,
    character $b$ immediately follows every crossing occurrence of $x$ in $T'$.
    Note that $xb$ is a prefix of $y$.
    Due to Lemma~\ref{lem:exist2}, $y$ occurs in $T$,
    implying $xb$ also occurs in $T$.
    Thus the number of right-extensions of $x$ in $T'$
    is no more than the number of right-extensions of $x$ in $T$.
    However, this contradicts $x \in \Qn$.
  \end{proof}

  \begin{lemma} \label{lem:sp35}
    For any $x \in \Ntwo \cup \Qn$ that is of Type $\rm{(iii)}$ or $\rm{(v)}$, 
    there do not exist $y,z \in \Ntwo \cup \Qn$ with
    $|y|>|x|$ and $|z|>|x|$
    satisfying $P_{x_L}=P_{{y}_{G}}$ and $P_{x_R}=P_{{z}_{F}}$ simultaneously,
    where $G,F \in \{L,R\}$.
  \end{lemma}

  \begin{proof}
    If $x_R$ is a suffix of $T'$, then clearly there is no $z$ satisfying
    $|z|>|x|$ and $P_{x_R}=S_{{z}_{G}}$.
    In what follows, we consider the case that $x_R$ is not a suffix of $T'$.
    
    Suppose that for $x \in \Ntwo \cup \Qn$ which is of Type $\rm{(iii)}$ or $\rm{(v)}$, 
    there exist $y,z \in \Ntwo \cup \Qn$ with $|y|>|x|$, $|z|>|x|$
    that satisfy $P_{x_L}=P_{{y}_{G}}$ and $P_{x_R}=P_{{z}_{F}}$ at the same time,
    where $G,F \in \{L,R\}$.
    Let $b$ and $d$~($b \neq d$) be the character immediately after $x_L$ in $T'$ and the character immediately after $x_R$ in $T'$, respectively.
    By Lemma~\ref{lem:exist2}, $y$ and $z$ occur in $T$, and hence $xb$ that is a prefix of $y$ and $xd$ that is a prefix of $z$ also occur in $T$.
    Therefore, $x \in \RightM(T)$. However, if $x \in \Ntwo(T)$, this contradicts $x \notin \RightM(T)$.
    Also, if $x \in \Qn$, the number of right-extensions of $x$ in $T'$ do not increase from the number of right-extensions of $x$ in $T$. However, this contradicts $x \in \Qn$.
  \end{proof}

  \subsection{Correspondence between $\Ntwo \cup \Qn$ and $\M(T)$}

  For any $x \in \Ntwo \cup \Qn$ that is of Type $\rm{(i)}$, $\rm{(ii)}$ or $\rm{(iv)}$, then we associate $x$ with both $P_{x_L}$ and $P_{x_R}$.
  For any $x \in \Ntwo \cup \Qn$ that is of Type $\rm{(iii)}$ or $\rm{(v)}$, 
  \begin{itemize}
    \item if there exists $y\in \Ntwo \cup \Qn$ with $|y|>|x|$
  such that $P_{x_R}=P_{{y}_{G}}$ where $G \in \{L,R\}$, then we associate $x$ with $P_{x_L}$ (see Figure~\ref{fig:pxrex});
    \item if there exists $y\in \Ntwo \cup \Qn$ with $|y|>|x|$ such that $P_{x_L}=P_{{y}_{G}}$ where $G \in \{L,R\}$, then we associate $x$ with $P_{x_R}$ (see Figure~\ref{fig:pxlex});
    \item otherwise, we associate $x$ with both $P_{x_L}$ and $P_{x_R}$.
  \end{itemize}

  \begin{figure}[H]
    \centering
    \includegraphics[keepaspectratio,scale=0.33]{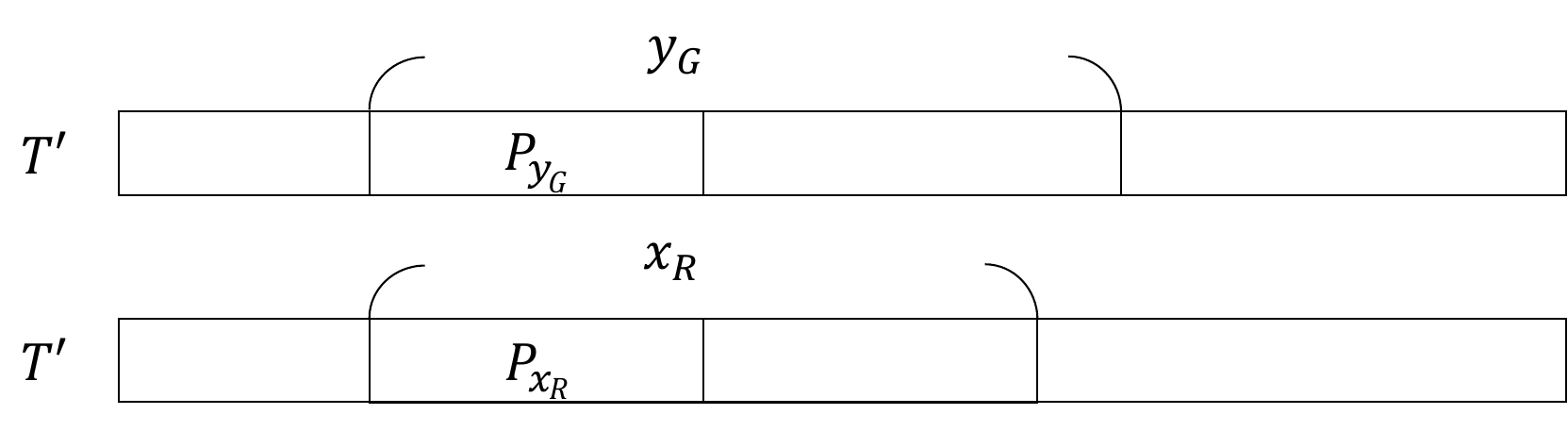}
    \caption{When there exists $y\in \Ntwo \cup \Qn$ with $|y|>|x|$
      such that $P_{x_R}=P_{{y}_{G}}$, where $G \in \{L,R\}$.}
    \label{fig:pxrex}
  \end{figure}

  \begin{figure}[H]
    \centering
    \includegraphics[keepaspectratio,scale=0.33]{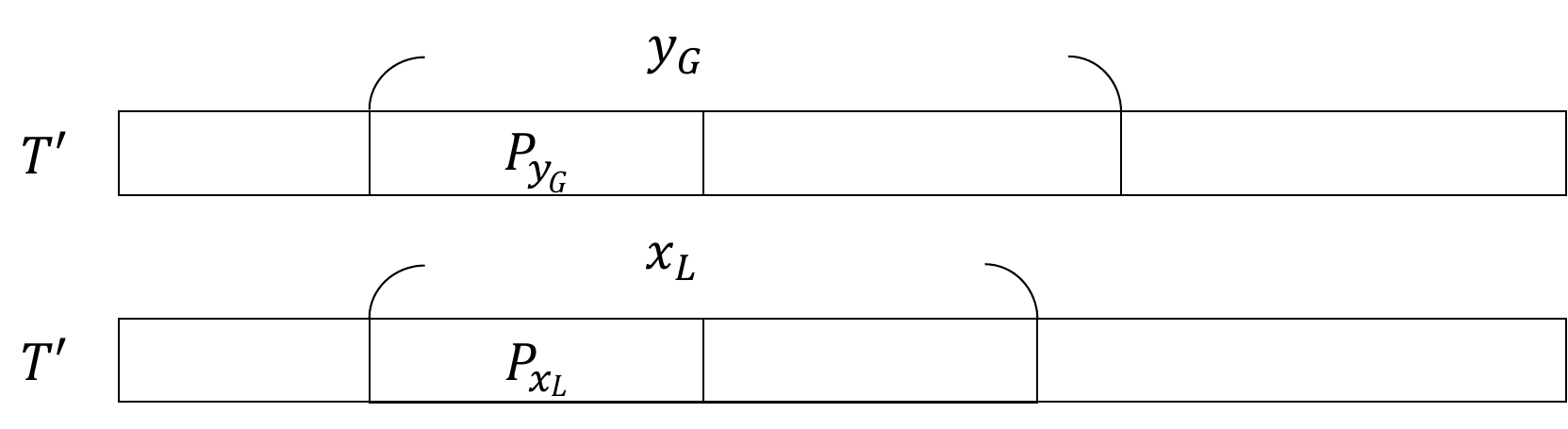}
    \caption{When there exists $y\in \Ntwo \cup \Qn$ with $|y|>|x|$ such that $P_{x_L}=P_{{y}_{G}}$, where $G \in \{L,R\}$.}
    \label{fig:pxlex}
  \end{figure}

  By Lemmas~\ref{lem:sp124} and~\ref{lem:sp35}, each $x \in \Ntwo \cup \Qn$ corresponds to a distinct string $P_{x_G}$, where $G \in \{L,R\}$.
  Below, for each $x \in \Ntwo \cup \Qn$,
  we define $H(x)$ and $I(x)$ to which $x$ corresponds:

  \begin{definition} \label{def:H_x_I_x}
    For each $x \in \Ntwo \cup \Qn$
    associated to ${P_{x_L}}$, let $H(x)=\rrep_T({P_{x_L}})$.
    For each $x \in \Ntwo \cup \Qn$
    associated to ${P_{x_R}}$, let $I(x)=\rrep_T({P_{x_R}})$.
    See Figure~\ref{fig:H_xI_x}.
    \sinote*{added}{%
    When there is only one crossing occurrence of $x$ (i.e. $x_L = x_R$),
    only $H(x)$ is defined as above and $I(x)$ is undefined.
    }%
  \end{definition}
$H(x)$ (resp. $I(x)$) is undefined
for any $x \in \Ntwo \cup \Qn$ that is \emph{not} associated to ${P_{x_L}}$
(resp. ${P_{x_R}}$).

\sinote*{added}{%
By Lemma~\ref{lem:exist2} every $x \in \Ntwo \cup \Qn$ occurs in $T$,
and thus $H(x)$ and $I(x)$ are well defined
when $x$ is associated to $P_{x_L}$ and $P_{x_R}$, respectively.
}%
  
%  \begin{definition} \label{def:H_x_I_x}
%    For any $x \in \Ntwo \cup \Qn$,
%    let $H(x)=\rrep_T({P_{x_L}})$.
%    let $I(x)=\rrep_T({P_{x_R}})$.
%  \end{definition}
%
%  Note that when we associate $x$ with only $P_{x_L}$ or $P_{x_R}$ for $x$, we only define $H(x)$ or $I(x)$ for $x$ and 
%  even if we only define $H(x)$ or $I(x)$ for $x$, we state both in follow proofs, but the proof is not incomplete.

  \begin{figure}[H]
    \centering
    \includegraphics[keepaspectratio,scale=0.33]{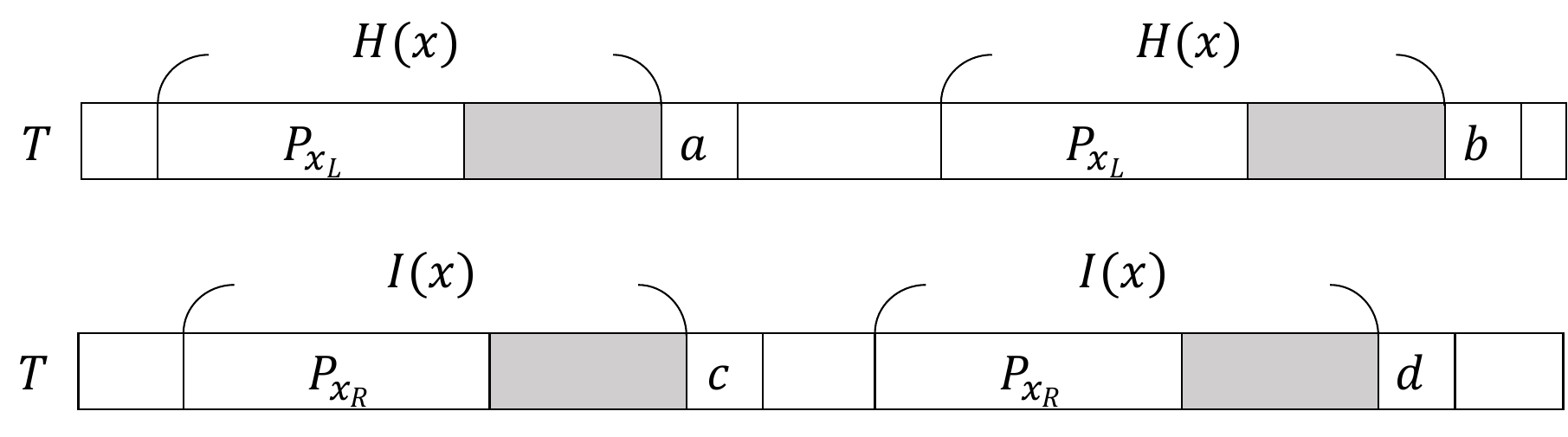}
    \caption{Illustration for $H(x)$ and $I(x)$ ($a\neq b, c\neq d$).}
    \label{fig:H_xI_x}
  \end{figure}

  \begin{lemma} \label{lem:H_x_I_x}
    For any $x \in \Ntwo \cup \Qn$, $H(x) \in \M(T)$ if $H(x)$ is defined,
    and $I(x) \in \M(T)$ if $I(x)$ is defined.
  \end{lemma}
  
%  \begin{lemma} \label{lem:H_x_I_x}
%    For any $x \in \Ntwo \cup \Qn$, $H(x), I(x) \in \M(T)$.
%  \end{lemma}

  \begin{proof}
    By Definition~\ref{def:H_x_I_x}, $H(x), I(x)\in \RightM(T)$.
    Therefore, it suffices for us to prove $H(x), I(x)\in \LeftM(T)$.
    For any $x \in \Ntwo \cup \Qn$, $x \in \LeftM(T)$.
    Since $P_{x_G} \: (G \in \{L,R\})$ is a prefix of $x$,
    we have $P_{x_G} \in \LeftM(T)$.
    Hence $H(x), I(x)\in \LeftM(T)$ holds.
  \end{proof}

  \begin{lemma} \label{lem:H_xH_y}
    For any $x,y \in \None \cup \Nnotv$ with $x \neq y$,
    let $\mathcal{L}$ be a list of $H(x)$, $I(x)$, $H(y)$, $I(y)$
    which are defined.
    Then the elements in $\mathcal{L}$ differ from each other.
  \end{lemma}

  \begin{proof}
    By a symmetrical argument to Lemma~\ref{lem:U_xU_y}.
  \end{proof}

  \subsection{Upper bound w.r.t. $\Ntwo \cup \Q$}

  \begin{lemma} 
    \label{lem:dt2}
    $\sum_{x \in \Ntwo \cup \Q}\D_{T'}(x) \le 3\size+2$.
  \end{lemma}

  \begin{proof}
    Below, we consider all the four possible cases depending on whether $x \in \Ntwo$ or $x \in \Qn$, and whether $H(x), I(x) \neq T$. 

    \noindent {\large \textbf{When $x\in \Ntwo$ and $H(x), I(x) \neq T$:}}
    \begin{itemize}
    \item
    First, we consider the case that $x$ is associated with both $H(x)$ and $I(x)$.
    Since $x \in \Ntwo$, then
      $x \notin \RightM(T)$.
    Therefore, the number of characters that are immediately after $x$ in $T$ is at most one.
    Moreover, there are at most two distinct characters immediately after the crossing occurrences of $x$.
    Hence, there are at most three distinct characters immediately after $x$ in $T'$, namely we have 
    \begin{equation}\label{equ:equ1}
      \D_{T'}(x) \le 3.
    \end{equation}
    In addition, since ${H(x),I(x) \ne T}$, it holds that $\D_{T}(H(x)),\D_{T}(I(x)) \ge 1$.
    By Inequality~\ref{equ:equ1}, we get $\D_{T'}(x) \le 3 \le \D_{T}(H(x))+\D_{T}(I(x))+1 \le 2\D_{T}(H(x))+2\D_{T}(I(x)).$

    \item Second, we consider the case that $x$ is associated with only one of $H(x)$ or $I(x)$.
    \begin{itemize}
    \item
    Assume that we associate $x$ with $H(x)$.
    Since $x \in \Ntwo$, then
      $x \notin \RightM(T)$.
    Therefore, the number of characters immediately after $x$ in $T$ is at most one.
    \rhnote*{added the case $x$ has only one crossing occurrence}{%
    In this case, we do not associate $x$ with $I(x)$, hence, $x$ has only one crossing occurrence or there exists $y\in \Ntwo \cup \Qn$ such that $|y|>|x|$ and $P_{x_R}=P_{{y}_{G}}$ where $G \in \{L,R\}$.
    When $x$ has only one crossing occurrence, there are at most one character immediately after the crossing occurrence of $x$.
    When there exists  $y\in \Ntwo \cup \Qn$ such that $|y|>|x|$ and $P_{x_R}=P_{{y}_{G}}$ where $G \in \{L,R\}$, then such $y$ occurs in $T$ due to Lemma~\ref{lem:exist2}.
    Therefore, although there are at most two distinct characters immediately after the crossing occurrences of $x$, one of them is the character immediately after $x$ in $T$ as shown in Figure~\ref{fig:xaex}.
    }%
    Hence, we have 
    \begin{equation}\label{equ:equ2}
      \D_{T'}(x) \le 2.
    \end{equation}
    In addition, since ${H(x) \ne T}$, then $\D_{T}(H(x)) \ge 1$ holds.
    By Inequality~\ref{equ:equ2}, we get $\D_{T'}(x) \le 2 \le \D_{T}(H(x))+1 \le 2\D_{T}(H(x)).$
    
    \item
    Let us assume that we associate $x$ with $I(x)$. In the same way as we associate $x$ with $H(x)$, we get $\D_{T'}(x) \le 2 \le \D_{T}(I(x))+1 \le 2\D_{T}(I(x))$.
    \end{itemize}
    \end{itemize}

    \begin{figure}[H]
      \centering
      \includegraphics[keepaspectratio,scale=0.33]{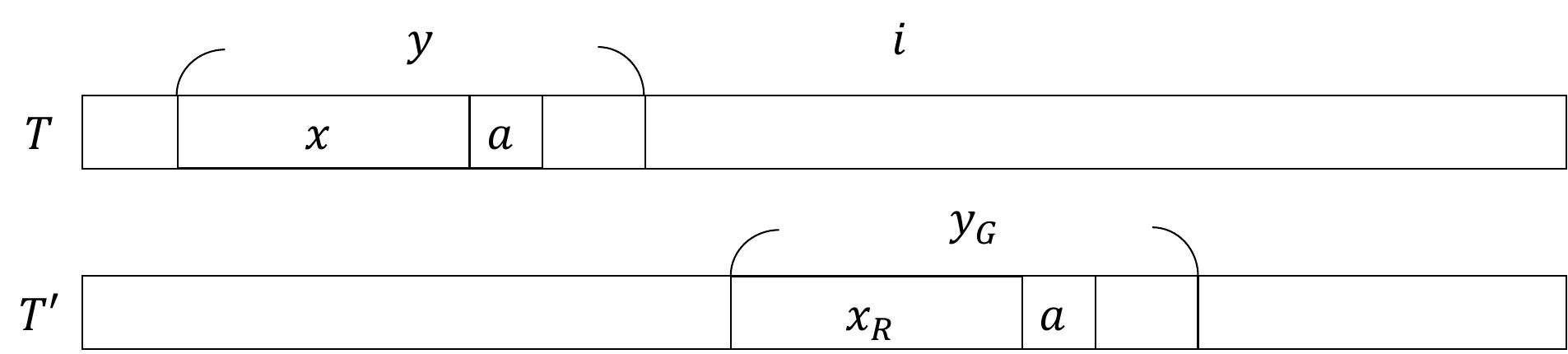}
      \caption{$xa$ occurs in $T$, where $a$ is the character immediately after the crossing occurrence $x_R$.}
      \label{fig:xaex}
    \end{figure}

    \noindent {\large \textbf{When $x\in \Ntwo$ and $H(x) = T$ or $I(x) = T$:}}
    \begin{itemize}
    \item
    Let $H(x) = T$. Now that $H(x)$ is defined, $x$ is associated with $H(x)$.
    \begin{itemize}
     \item First, we consider the case that we associate $x$ with both $H(x)$ and $I(x)$.
    In the same way as in Inequality~\ref{equ:equ1}, we get $\D_{T'}(x) \le 3$.
    In addition, since $H(x)=T$ and Lemma~\ref{lem:H_xH_y} holds, $I(x) \neq T$ and thus $\D_{T}(I(x)) \ge 1$ holds.
    Hence, we have $\D_{T'}(x) \le 3 \le \D_{T}(I(x))+2 \le 2\D_{T}(I(x))+2 \le 2\D_{T}(H(x))+2\D_{T}(I(x))+2$.
    \item Second, we consider the case that we only associate $x$ with $H(x)$.
    Since $H(x)=T$, $\D_{T}(H(x)) =0$ holds.
    In the same way as in Inequality~\ref{equ:equ2}, we get $\D_{T'}(x) \le 2$.
    Thus, we have $\D_{T'}(x) \le 2 \le 2\D_{T}(H(x))+2$.
    \end{itemize}

    \item
    Let $I(x) = T$. Now that $I(x)$ is defined, $x$ is associated with $I(x)$.
    In the same way as in the case for $H(x)=T$, 
    we get $\D_{T'}(x) \le 3 \le \D_{T}(H(x))+2 \le 2\D_{T}(H(x))+2 \le 2\D_{T}(H(x))+2\D_{T}(I(x))+2$ in the case that we associate $x$ with both $H(x)$ and $I(x)$,
    and we get $\D_{T'}(x) \le 2 \le 2\D_{T}(I(x))+2$ in the case that we only associate $x$ with $I(x)$.
   \end{itemize}

%   \medskip 
   \noindent {\large \textbf{When $x\in \Qn$ and $H(x), I(x) \neq T$:}}    
    Here, we analyze $\D_{T'}(x)-\D_T{(x)}$ since $x \in \Qn$.
    \begin{itemize}
    \item
    First, we consider the case that we associate $x$ with both $H(x)$ and $I(x)$.
    There are at most two distinct characters immediately after the crossing occurrences of $x$.
    Hence,
    \begin{equation}\label{equ:equ3}
      \D_{T'}(x)-\D_T{(x)} \le 2.
    \end{equation}
    In addition, since ${H(x),I(x) \ne T}$, then $\D_{T}(H(x)),\D_{T}(I(x)) \ge 1$ holds.
    By Inequality~\ref{equ:equ3}, we get $\D_{T'}(x)-\D_T{(x)} \le 2 \le \D_{T}(H(x))+\D_{T}(I(x)) \le \D_{T}(H(x))+\D_{T}(I(x))$.

    \item
    Second, we consider the case that we associate $x$ with only one of $H(x)$ or $I(x)$.
    Here, let us assume that we associate $x$ with $H(x)$.
    \rhnote*{added the case $x$ has only one crossing occurrence}{%
    In this case, we do not associate $x$ with $I(x)$, hence, $x$ has only one crossing occurrence or there exists $y\in \Ntwo \cup \Qn$ such that $|y|>|x|$ and $P_{x_R}=P_{{y}_{G}}$ where $G \in \{L,R\}$.
    When $x$ has only one crossing occurrence, there are at most one character immediately after the crossing occurrence of $x$.
    When there exists  $y\in \Ntwo \cup \Qn$ such that $|y|>|x|$ and $P_{x_R}=P_{{y}_{G}}$ where $G \in \{L,R\}$, then such $y$ occurs in $T$ due to Lemma~\ref{lem:exist2}.
    Therefore, although there are at most two distinct characters immediately after the crossing occurrences of $x$, one of them is the character immediately after $x$ in $T$ as shown in Figure~\ref{fig:xaex}.
    }%
    Hence, we have
    \begin{equation}\label{equ:equ4}
      \D_{T'}(x)-\D_T{(x)} \le 1.
    \end{equation}
    In addition, since ${H(x) \ne T}$, then $\D_{T}(H(x)) \ge 1$ holds.
    By Inequality~\ref{equ:equ4}, we get $\D_{T'}(x)-\D_T{(x)} \le 1 \le \D_{T}(H(x))$.
    In the case that we associate $x$ with $I(x)$, in the same way as we associate $x$ with $H(x)$, we get $\D_{T'}(x)-\D_T{(x)} \le 1 \le \D_{T}(I(x))$.
    \end{itemize}

%   \medskip 
   \noindent {\large \textbf{When $x\in \Qn$ and $H(x) = T$ or $I(x) = T$:}}
   Here, we analyze $\D_{T'}(x)-\D_T{(x)}$ since $x \in \Qn$.
   \begin{itemize}
    \item
    Let $H(x) = T$. Since $H(x)$ is defined, $x$ is associated with $H(x)$.
    \begin{itemize}  
     \item First, let us consider the case that we associate $x$ with both $H(x)$ and $I(x)$.
    In the same way as in Inequality~\ref{equ:equ3}, we get $\D_{T'}(x)-\D_T{(x)} \le 2$.
    In addition, since $H(x)=T$ and Lemma~\ref{lem:H_xH_y} holds, $I(x) \neq T$ and thus $\D_{T}(I(x)) \ge 1$ holds.
    Hence $\D_{T'}(x)-\D_T{(x)} \le 2 \le \D_{T}(I(x)) + 1 \le \D_{T}(H(x)) + \D_{T}(I(x)) + 1$.

     \item Second, let us consider the case that we only associate $x$ with $H(x)$.
    Since $H(x)=T$, $\D_{T}(H(x))=0$ holds.
    In the same way as in Inequality~\ref{equ:equ4}, we get $\D_{T'}(x)-\D_T{(x)} \le 1$.
    Thus, we have $\D_{T'}(x)-\D_T{(x)} \le 1 \le \D_{T}(H(x))+1$.
    \end{itemize}

   \item
    Let $I(x) = T$. Since $I(x)$ is defined, $x$ is associated with $I(x)$.
    In the same way as in the case for $H(x)=T$, 
    we get $\D_{T'}(x)-\D_T{(x)} \le 2 \le \D_{T}(H(x)) + 1 \le \D_{T}(H(x)) + \D_{T}(I(x)) +1$ in the case that we associate $x$ with both $H(x)$ and $I(x)$,
    and we get $\D_{T'}(x)-\D_T{(x)} \le 1 \le \D_{T}(I(x))+1$ in the case that we only associate $x$ with
    \rhnote*{deleted "only one of $H(x)$"}{%
    $I(x)$.
    }%
   \end{itemize}
    
    \begin{table}[h]
      \centering
      \caption{Upper bounds for each case of Lemma~\ref{lem:dt2}.} 
      \label{inequality}
      \fontsize{9pt}{10pt}\selectfont
      \begin{tabular}{|c|c|c|} \hline
        & When $H(x) \ne T \land I(x) \ne T$ & When $H(x) = T \lor I(x) = T$ \\ \hline
        $\D_{T'}(x)-\D_T{(x)} \: (x \in \Qn)$ & $\leq \D_T{(H(x))}+\D_T{(I(x))}$ & $\leq \D_T{(H(x))}+\D_T{(I(x))}+1$ \\\hline
        $\D_{T'}(x) \: (x \in \Ntwo)$ & $\leq 2(\D_T{(H(x))}+\D_T{(I(x))})$ & $\leq 2(\D_T{(H(x))}+\D_T{(I(x))})+2$ \\ \hline
      \end{tabular}
    \end{table}

%  For simplicity, we state both $H(x)$ and $I(x)$ even if we only associate $x$ with $H(x)$ or $I(x)$.
  %  For instance, if we only associate $x$ with $H(x)$, let $\D_T{(I(x))}$ be zero.

  \noindent {\large \textbf{Wrapping up:}}    
  Table~\ref{inequality} summarizes the bounds obtained above.
  For simplicity,
  let $\D_T{(I(x))} = 0$ when $I(x)$ is undefined,
  and let $\D_T{(H(x))} = 0$ when $H(x)$ is undefined.
  Note that this does not affect our upper bound analysis,
  since no maximal repeats in $T'$ are associated to the undefined $H(x)$'s and $I(x)$'s.
  By Lemma~\ref{lem:H_xH_y}, there is at most one string $x$ such that $H(x) = T$ or $I(x) = T$.
  Thus, by using Lemma~\ref{lem:H_xH_y} and summing up the values in Table~\ref{inequality}, 
  we obtain $\sum_{x\in \Ntwo}\D_{T'}(x) + \sum_{x\in \Qn}(\D_{T'}(x)-\D_T(x)) \le 2\size+2$.
  Also, since the number of out-edges of $x \in \Qnotn$ does not increase,
  we get
  $\sum_{x\in \Qnotn}\D_{T'}(x) + \sum_{x\in \Qn}\D_{T}(x) \le 
  \sum_{x\in \Qnotn}\D_{T}(x) + \sum_{x\in \Qn}\D_{T}(x) \le
  \sum_{x\in \Q}\D_{T}(x) \le \sum_{x\in \M(T)}\D_{T}(x) = \size$.
  By adding $\sum_{x\in \Ntwo}\D_{T'}(x) + \sum_{x\in \Qn}(\D_{T'}(x)-\D_T(x)) \le 2\size+2$, 
  we get $\sum_{x \in \Ntwo \cup \Q}\D_{T'}(x) \le 3\size+2$.
  \end{proof}

\section{Upper bound for total out-degrees of nodes w.r.t. $\Nv$}
\label{sec:upper_bound_3}

In this section, we show an upper bound for the total out-degrees of nodes corresponding to strings that are elements of $\Nv  \subseteq \M(T')$.

We first describe useful properties of
strings $x \in \Nv$.

\begin{definition}
    For any $x \in \Nv$, let $J_x$ be the string that is obtained by removing $P_{x_R}$ and $S_{x_L}$ from $x$, namely $x = P_{x_R} J_x S_{x_L}$.
  \end{definition}

Note that, by the definition of Type $\rm{(v)}$,
each $x \in \Nv$ has two or more crossing occurrences in $T'$.
Hence $J_x$ always exists (possibly the empty string).
See Figures~\ref{fig:J_x} and~\ref{fig:J_x2}.

  \begin{figure}[H]
    \centering
    \includegraphics[keepaspectratio,scale=0.35]{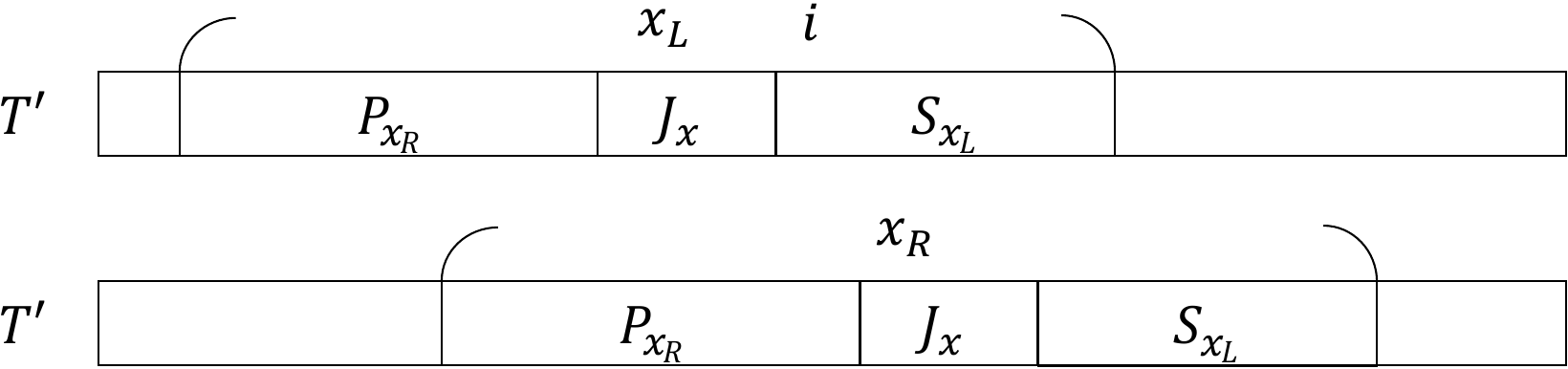}
    \caption{Illustration for $J_x$ in case of insertions and substitutions.
      %      In the case of deletion, the right-end of $J_x$ in $x_L$ touches the left-end of $j_x$ in $x_R$.
    }
    \label{fig:J_x}
  \end{figure}

  \begin{figure}[H]
    \centering
    \includegraphics[keepaspectratio,scale=0.35]{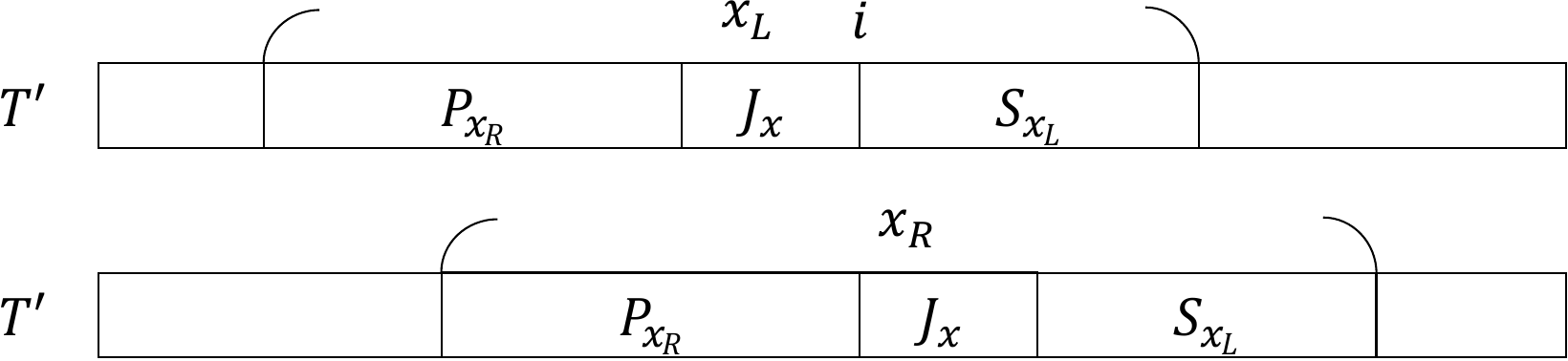}
    \caption{Illustration for $J_x$ in case of deletions.}
    \label{fig:J_x2}
  \end{figure}

  \begin{lemma}\label{lem:J_xJ_y}
    For any $x, y \in \Nv$ with \sinote*{added}{$x \neq y$}, $J_x \neq J_y$.
  \end{lemma}

  \begin{proof}
    For a contrary, suppose that there exist $x, y \in \Nv$ such that $x \neq y$ and $J_x = J_y$.
%    and assume withoug loss of generality that $|S_{x_{L}}| \leq |S_{y_{L}}|$.
    Since $x$ is of Type $\rm{(v)}$, the characters immediately after $x_L$ and $x_R$ are different and let $a, b$~($a \neq b$) be these characters, respectively.
    If $|S_{x_{L}}|<|S_{y_{L}}|$, then both $S_{x_{L}}a$ and $S_{x_{L}}b$ must be prefixes of $S_{y_{L}}$ (see Figures~\ref{fig:J_xJ_y} and~\ref{fig:J_xJ_y2}), which contradicts that $a \neq b$.
    The other case where $|S_{x_{L}}| > |S_{y_{L}}|$ also leads to a contradiction.
    Hence $S_{x_{L}}=S_{y_{L}}$.
    Also, $P_{x_{R}}=P_{y_{R}}$ follows in a symmetric manner.
    These imply $x=y$, which is a contradiction.
  \end{proof}

  \begin{figure}[H]
    \centering
    \includegraphics[keepaspectratio,scale=0.35]{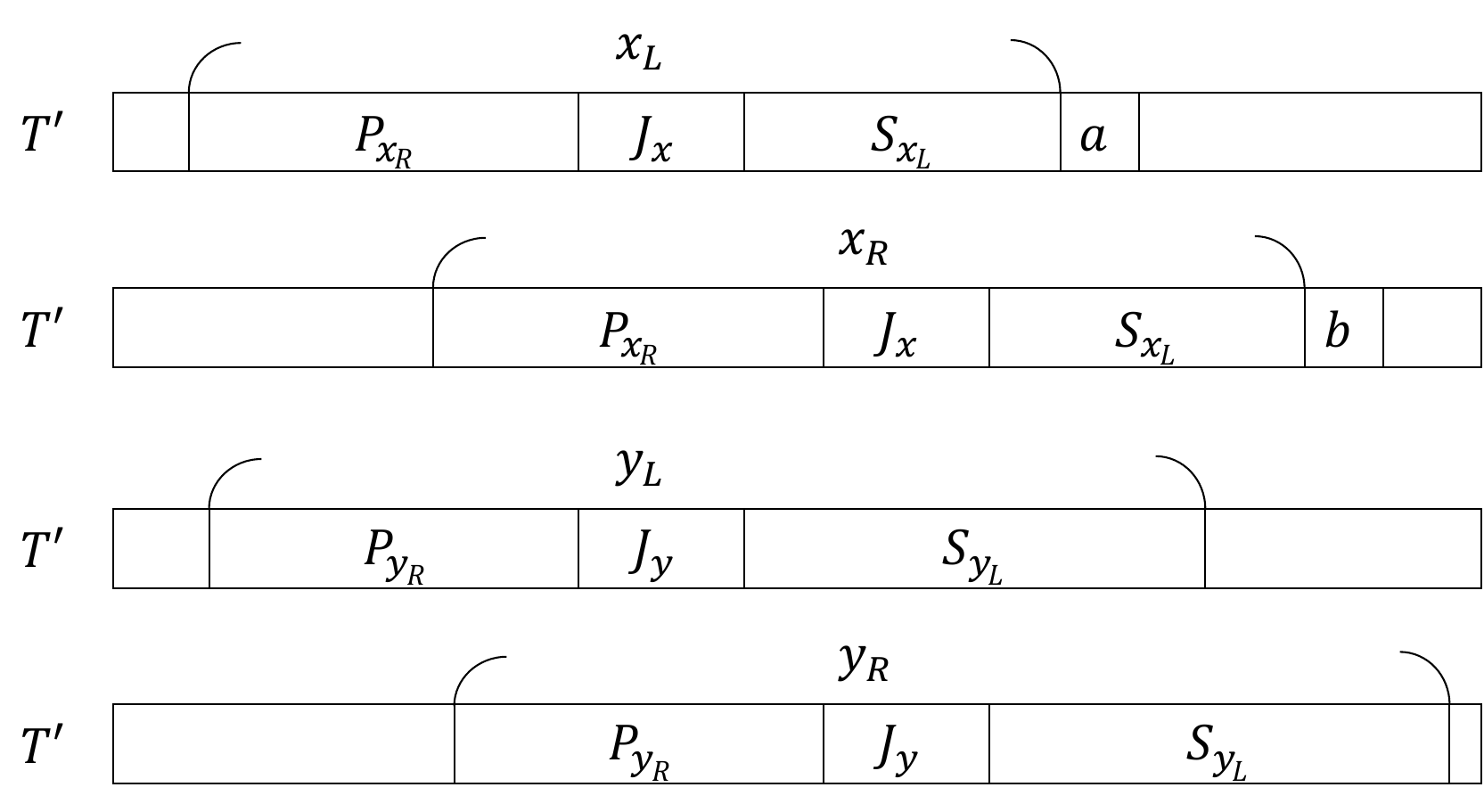}
    \caption{Illustration for Lemma~\ref{lem:J_xJ_y} in case of insertions and substitutions, where $J_x = J_y$.
      %      In the case of deletion, the right-end of $J_x$ in $x_L$ touches the left-end of $J_x$ in $x_R$, and the right-end of $J_y$ touches the left-end of $J_y$ in $y_R$.
    }
    \label{fig:J_xJ_y}
  \end{figure}

  \begin{figure}[H]
    \centering
    \includegraphics[keepaspectratio,scale=0.35]{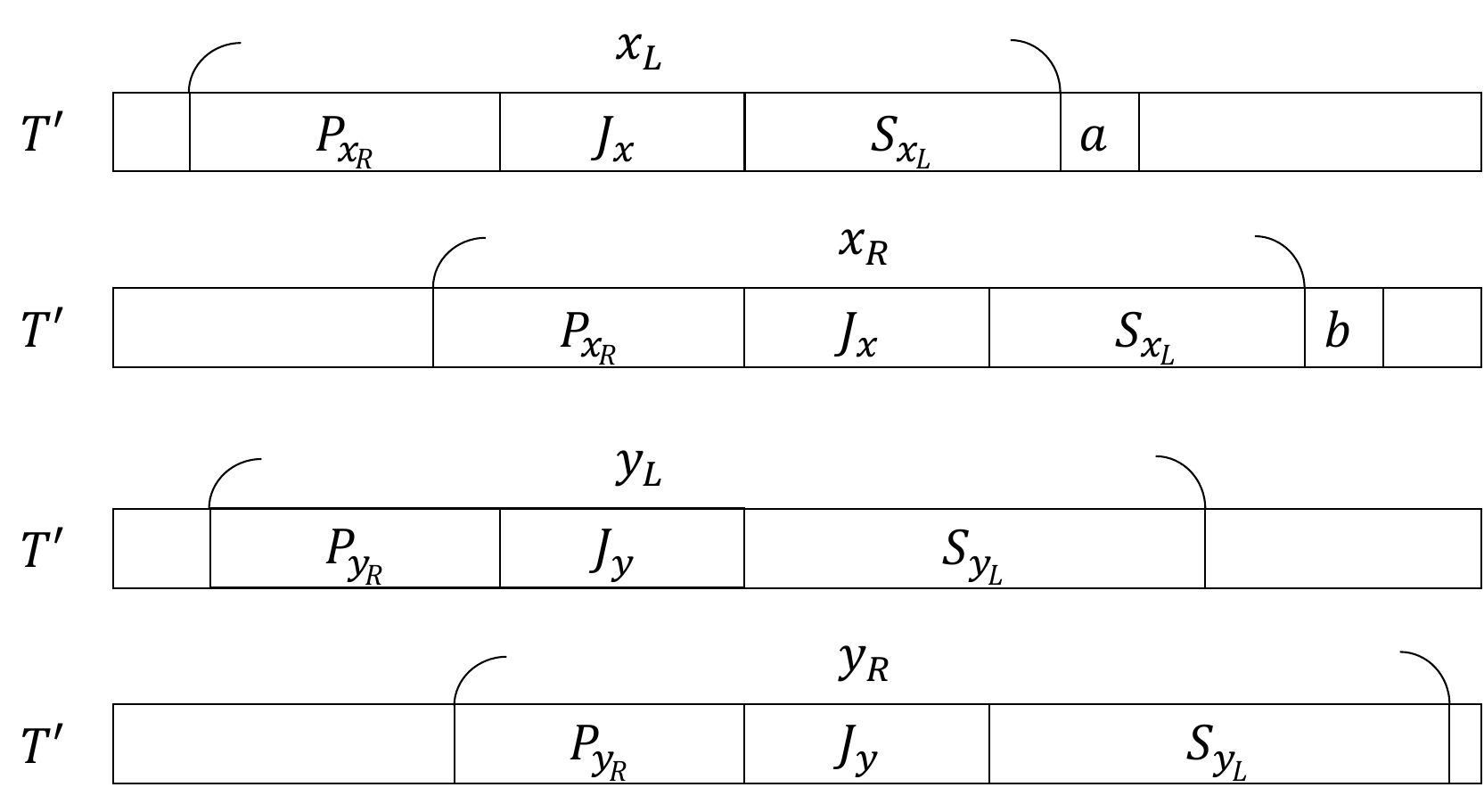}
    \caption{Illustration for Lemma~\ref{lem:J_xJ_y} in case of deletions, where $J_x = J_y$.}
    \label{fig:J_xJ_y2}
  \end{figure}

  \subsection{Correspondence between $\Nv$ and $\M(T)$}
  For any $x \in \Nv$, we associate $x$ with $J_x$.
  For any $x \in \Nv$
  we define $K(x)$ to which $x$ corresponds, by using $J_x$,
  as follows:

  \begin{definition}\label{def:K_x}
    For any $x \in \Nv$, let $K(x)=\lrep_T({\rrep_T({J_x})}) = \rrep_T({\lrep_T({J_x})})$ (see also Figure~\ref{fig:K_x} for illustration).
  \end{definition}
  We note that $K(x)$ is well defined since $J_x$ is a substring of $T$.

  \begin{figure}[H]
    \centering
    \includegraphics[keepaspectratio,scale=0.35]{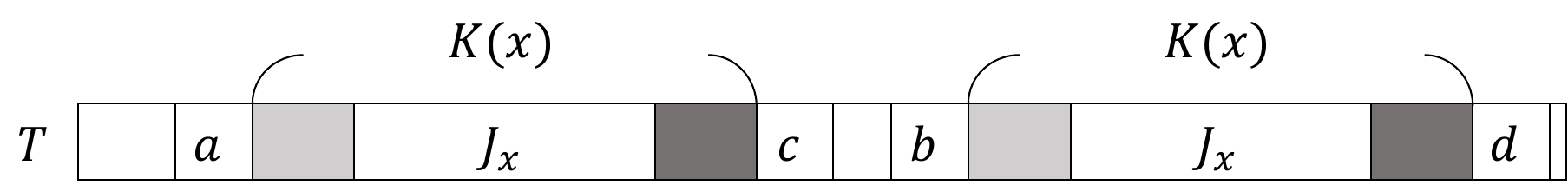}
    \caption{Illustration for Definition~\ref{def:K_x}, where $a\neq b$ and $c\neq d$.}
    \label{fig:K_x}
  \end{figure}

  \begin{lemma}\label{lem:K_xK_y}
    For any $x, y \in \Nv$ \sinote*{added}{with $x \neq y$}, $K(x)\neq K(y)$.
  \end{lemma}

  \begin{proof}
    Suppose that there exist $x, y \in \Nv$ such that $x \neq y$ and $K(x)= K(y)$,
    and without loss of generality that $|J_x| \leq |J_y|$.
    Then, $J_x$ occurs at least twice in $J_y$.
    Therefore, $K(x) \neq K(y)$, however, this is a contradiction.
  \end{proof}

  \subsection{Upper bound w.r.t. $\Nv$}

  \begin{lemma} 
    \label{lem:dt3}
    $\sum_{x \in \Nv}\D_{T'}(x) \le 2\size$.
  \end{lemma}
  
  \begin{proof}
    By the definition of $\Nv$,
    there is no other right-extension of $x$ in $T'$ than the right-extension(s) of the crossing occurrence(s) of $x$.
    %    Now, since there are at most two distinct characters immediately after the crossing occurrences of $x$,
    Thus 
    there are at most two distinct characters immediately after $x$ in $T'$.
    Since $K(x)$ occurs at least twice in $T$, $\D_T(K(x)) \geq 1$.
    Hence, we get $\D_{T'}(x) \leq 2 \leq 2\D_T(K(x))$.
    Therefore, we have $\sum_{x \in \Nv}\D_{T'}(x) \le 2\size$ by Lemma~\ref{lem:K_xK_y}.
  \end{proof}

\section{Conclusions}
This paper proved that the worst-case multiplicative sensitivity
of CDAWGs is asymptotically at most 8.
Our analysis is based fully on new combinatorial properties
of maximal repeats and their right-extensions,
that are incidental to an edit operation on the strings.
The only known lower bound for the multiplicative sensitivity of CDAWGs
is asymptotically 2~\cite{FujimaruNI25}.
It is intriguing future work to close the gap between
these upper bound and lower bound.

\bibliographystyle{abbrv}
\bibliography{ref}

\end{document}